%% file: article.tex
\documentclass[sigconf,balance=true]{acmart} \input{preamble}

\author{Louis Gaillard} \affiliation{%
  \institution{ENS de Lyon, CNRS, Inria,
    Université Claude Bernard Lyon 1, LIP, UMR 5668 }
  \city{Lyon}
  \country{France}
}
\renewcommand{\shortauthors}{L. Gaillard} \title[
A unified approach
for degree bound estimates of linear differential
operators]{A unified approach
  for degree bound estimates \\ of linear differential
  operators}

 \copyrightyear{2025} \acmYear{2025}
\acmConference[ISSAC
 '25]{International Symposium on Symbolic and Algebraic
 Computation}{2025}{Guanajuato, Mexico}
 \acmBooktitle{International Symposium on Symbolic and Algebraic
 Computation (ISSAC '25), 2025, Guanajuato, Mexico}

\begin{abstract}
  
  We identify a common scheme in
  several existing algorithms
  addressing computational problems on
  linear differential equations with polynomial coefficients.
  These algorithms reduce to computing a linear relation between vectors
  obtained as iterates of a simple differential operator known as
  \emph{pseudo-linear map}.

  We focus on establishing precise degree bounds on the output of this
  class of algorithms. It turns out that in all known instances
  (least common left multiple, symmetric product,\dots), the bounds that are
  derived from the linear algebra step using Cramer's rule are pessimistic. The gap with
  the behaviour observed in practice is often of one order of magnitude, and better
  bounds are sometimes known and derived from \emph{ad hoc} methods and independent arguments.
  We propose a unified approach for proving output degree bounds for all instances of
  the class at once.
  The main technical tools come
  from the theory of realisations of matrices of rational functions
  and their determinantal denominators.

\end{abstract}

\begin{document}



\maketitle

\section{Introduction}
\label{sec:introduction}

In this work, we study four central algorithmic
problems involving \emph{D-finite}
functions~\cite{stanley1980dfinite,Kauers23},
which are defined by linear differential equations with polynomial
coefficients.
We specifically focus on the Hermite
reduction-based creative telescoping algorithm for the simple
integration of bivariate rational
functions~\cite{bostan2010complexity}, the computation
of a differential equation for an algebraic function~\cite{bostan2007algeqtodiffeq}
and two closure operations on D-finite functions: least common
left multiple~\cite{vanderHoeven2016complexity_skew_arith,bostan2012LCLM}
and symmetric product~\cite{kauers2014bounds_Dfinite_closure,stanley1980dfinite}.
%\louis {check references}
These four problems can all be seen as particular instances
of a large class of algorithmic problems that share the
common property to be
expressible with the following algebraic formulation.
Let $k$ be a field of
characteristic 0 and $\DP_x$ be the usual derivation on the
rational function field $k(x)$. We are given a
sequence $(a_i)_{i\ge 0}$ of $k(x)^n$ defined by the
simple inductive formula:
\begin{align}
  \label{eq:def_ai}
  a_0 \in k[x]^n, \quad a_{i+1} = \left( \DP_x + T \right) \cdot a_i,~ i \ge 0,
\end{align}
where $T \in k(x)^{n \times n}$ is a $k(x)$-linear map. % and $\DP_x$
% denotes the derivation with respect to $x$.
The problem is to compute 
the coefficients $\eta_i \in k[x]$ of a minimal
linear relation
$\eta_\rho(x) \cdot a_\rho + \cdots + \eta_0(x) \cdot a_0 = 0$.
In~(\ref{eq:def_ai}), the operator $\theta = \DP_x + T$ is an example of a
\emph{pseudo-linear
  map}~\cite{jacobson1937pseudolineartransformations,bronstein1996introduction}, and
we have $a_i = \theta^i a_0$.
In other words, we consider the following problem.
\begin{problem}
  \label{pb:lin_rel_pseudo_lin_map}
  Given $T \in k(x)^{n \times n}$, $a \in k[x]^n$ and letting
  $\theta = \DP_x + T$ and
  $\rho = \dim\spn_{k(x)}(\theta^i a, i \ge 0) \le n$, find
  $\eta = (\eta_0,\dots, \eta_{\rho}) \in k[x]^{\rho+1} \setminus \{ 0 \}$ such that
  \begin{equation}
    \label{eq:linear_relation_pseudo_linear_map}
    \eta_\rho \cdot \theta^\rho a + \cdots + \eta_1 \cdot \theta a + \eta_0 \cdot a = 0.
  \end{equation}
\end{problem}
\noindent In the case $\rho = n$, \ie $a$ is a cyclic vector
for
$T$~\cite{churchill2002cyclicvectors,bostan2013cyclic,adjamagbo1988cyclic,barkatou1993algo_uncoupling},
the above relation
gives a characteristic polynomial of $\theta$~\cite{amitsur1954diffpoly_divalgebra}.
Problem~\ref{pb:lin_rel_pseudo_lin_map} can be seen as an extension
of the method of Krylov iterates
for computing the characteristic
polynomial~\cite{neiger_pernet_villard_2024_krylov_iterates,keller-gehrig_1985_charpoly}
in the differential context.
As further developed in Section~\ref{sec:applications}, each of the four
problems mentioned above reduces to solving
the linear system~(\ref{eq:linear_relation_pseudo_linear_map})
with a specific choice of $T$ and $a$.

% % \paragraph*{Problem statement and motivations}
% This work is initiated by an observation
% that is common to several algorithmic problems involving \emph{D-finite}
% functions~\cite{lipshitz1989dfinite,stanley1980dfinite},
% \ie defined by linear differential equations with polynomial
% coefficients.
% Let $k$ be a field of
% characteristic 0 and $\DP_x$ be the usual derivation on the
% rational function field $k(x)$.
% For several problems, the computed
% linear differential operator in $k[x]\left<\DP_x\right>$ is
% derived from the minimal $k(x)$-linear relation
% between the first terms of a sequence $(a_i)_{i \ge 0}$ of $k(x)^n$
% defined by a simple inductive formula:

% Up to normalization of
% the rational functions, the coefficients of the desired equation are
% directly given from the minimal linear relation
% $\eta_\rho(x) \cdot a_\rho + \cdots + \eta_0(x) \cdot a_0 = 0$ for
% $\eta_i(x) \in k[x]$.

% For example, several
% closure operations on \emph{D-finite} functions
% (least common left multiple~\cite{bostan2012LCLM}, symmetric
% product~\cite{kauers2014bounds_Dfinite_closure},
% algebraic substitution~\cite{bostan2007algeqtodiffeq}) or the Hermite
% reduction-based creative telescoping algorithm for the simple
% integration of bivariate rational
% functions~\cite{bostan2010complexity} can be seen as particular instances
% of Problem~\ref{pb:lin_rel_pseudo_lin_map}.
% We first describe these operations in more detail.
% % \louis{Present the applications}

\paragraph*{Hermite reduction-based creative telescoping}

Creative telescoping refers to a family of methods that have proven successful
for symbolic integration and summation.  The general
concept was first introduced by
Zeilberger~\cite{zeilberger1990creativetelesc} and it applies to a
large class of functions.  Here we focus on the simple
integration of a bivariate rational function
$f(x,y) = p /q \in k(x,y) \setminus \{ 0 \}$ with $p$ and $q$ coprime.  
For any integration domain $\omega$, the definite integral
$F(x) = \int_\omega f(x,y) \, dy$ 
is known to be D-finite and the problem consists in computing a linear differential equation
satisfied by $F$.  The heart of the method relies on finding a
\emph{telescoper} for~$f$, namely a linear differential operator
$L \in k[x]\left<\DP_x\right> \setminus \{ 0 \}$ such that
$L(x,\DP_x) \cdot f(x,y) = \DP_y \cdot h(x,y)$ for some rational function
$h$. Under reasonable assumptions on~$h$ and~$\omega$, the operator $L$ is then proved
to annihilate $F$.
Performing Hermite
reduction~\cite{hermite1872integration} on $\DP_x^i(f)$
iteratively enables to compute the coefficients of $L$ as a
solution to
Problem~\ref{pb:lin_rel_pseudo_lin_map}~\cite{bostan2010complexity}.

% Next, for any rational function
% $g \in k(x,y)$ whose denominator is a power of $q$, its Hermite
% reduction $\herm(g)$~\cite{hermite1872integration} is defined as the
% unique polynomial $r \in k(x)[y]$ such that $g = \DP_y(h) + r/q$ for
% $h \in k(x,y)$ and $\deg_{y}(r) < \deg_y(q)$.  As a consequence,
% $\herm(g) = 0$ if and only if $g$ is a derivative (wrt $y$).  Thus, a
% telescoper $L$ for $f$ is such that $\herm(L\cdot f) = 0$. So by
% linearity, it can be seen as a linear relation between
% $\herm(\DP_x^i \cdot f)$ for $i \ge 0$.  Finally, one can write
% $\herm(\DP_x^i \cdot f)= \theta^i(p)$ with $p$ the numerator of $f$
% and $\theta$ a pseudo-linear map $\DP_x + T$.  Therefore, a minimal
% order telescoper for $f$ can be expressed as a solution of a
% particular instance of
% Problem~\ref{pb:lin_rel_pseudo_lin_map}~\cite[Sec.~3.1]{bostan2010complexity}.

\paragraph*{Differential equation for an algebraic function}
A power series $\alpha(x)$ over the field $k$ is said to be \emph{algebraic}
if it is a root of a nonzero bivariate polynomial $P(x,y) \in k[x,y]$.
It is known that such an algebraic function is also
D-finite~\cite[p.~287]{abeloeuvres},
and we call \emph{differential resolvent}
an operator annihilating all the roots of $P$.
Cockle's algorithm~\cite{cockle1861transcendental,bostan2007algeqtodiffeq}
computes a differential resolvent
as a linear relation between the images %$D_i$
of $\alpha^{(i)}$
in the quotient ring $k(x)[y]/(P)$ that is a finite dimensional
vector space over $k(x)$.
% Moreover, one can
% write $D_{i+1} = \theta (D_i)$ with $\theta = \DP_x + T$ a
% pseudo-linear map~\cite[Sec.~2.1]{bostan2007algeqtodiffeq}.
% Hence Cockle's algorithm
% solves an instance of Problem~\ref{pb:lin_rel_pseudo_lin_map}.

\paragraph*{Least common left multiple and symmetric product}
D-finite functions are closed under sums and
products, and the algorithms
derived from the proof are expressible as instances
of Problem~\ref{pb:lin_rel_pseudo_lin_map}~\cite[Thm.~2.3]{stanley1980dfinite}.
Efficient algorithms for these two operations are
needed for manipulating D-finite functions represented by
annihilating operators~\cite{salvy2019survey,salvy1994gfun}. 
For $L_1,\dots,L_s \in k[x]\left<\DP_x\right>$, a minimal
order linear differential operator
that annihilates all the sums (resp.~products)
of solutions of the $L_i$'s is called a
least common left
multiple (LCLM) (resp.~symmetric product).

\vspace{5pt}
Our thesis is that studying Problem~\ref{pb:lin_rel_pseudo_lin_map}
can contribute to a deeper understanding
of the four specific instances mentioned above, since
this common formulation underlines a shared structure between these problems.
Following this idea, we specifically propose a unified approach to prove
better degree bounds on the polynomials $\eta_i$ in~\eqref{eq:linear_relation_pseudo_linear_map}
(or to retrieve the best known
bounds)
for these four instances at once.
Such bounds are useful as they provide estimates on the output size of algorithms
addressing these problems, and thus determine the notion of optimal algorithms in an
algebraic model of complexity. In addition, there are examples
for algorithms whose running time
is dependent on the output
size~\cite{vanderHoeven2016complexity_skew_arith,bostan2007algeqtodiffeq}.
In this setting, proving tight degree bounds is of key importance for
complexity analysis.

% On the one hand, it
% yields good estimates on the size of the output. Thus, it provides a
% lower-bound on the complexity of the problem that makes it possible to
% determine whether a given algorithm is \emph{optimal} in the algebraic
% model of complexity.  On the other hand, the precise complexity
% analysis of an algorithm often requires tight degree bounds on the
% numerous objects computed in the core of the algorithm.  Finally, it
% can also lead to a better understanding of the objects involved and
% give new algorithmic ideas and perspectives.
% % \louis{Explain why we are interested in degree bounds in general}

A straightforward approach to obtain degree bounds is to consider~$\eta$
as the solution of the linear system defined
in~(\ref{eq:linear_relation_pseudo_linear_map}) and to derive degree
bounds using classical techniques such as Cramer's
rule~\cite{bostan2010complexity,kauers2014bounds_Dfinite_closure}.
By doing so, if $D$ is a degree bound on both the numerator and the
denominator of $T$, one can show that the coefficients of the solution
of Problem~\ref{pb:lin_rel_pseudo_lin_map} have degrees bounded by
$O(n\deg a + n^2 D)$.  The precise result is
stated and proved in Section~\ref{sec:direct_bound}.

Experiments suggest that this bound is reached for \emph{generic}
instances of Problem~\ref{pb:lin_rel_pseudo_lin_map}.  However for
our four instances, in which the matrix $T$ possesses a structure,
the bound above is an overestimation of the actual degrees.
For instance, let us consider the case of the algorithm based on Hermite's reduction
applied to a rational function $f$ with $d_x$ (resp.~$d_y$) a degree bound
in $x$ (resp.~$y$) on both its numerator and denominator.
By applying Cramer's rule in  the linear algebra step
of the algorithm,
the degree of the telescoper is bounded by
$O(d_y^3 d_x)$~\cite[Lem.~18]{bostan2010complexity}.
However, an independent analysis based on the Almkvist-Zeilberger
algorithm for rational functions~\cite{almkvist1990method} yields a better
bound in $O(d_y^2 d_x)$~\cite[Thm.~25]{bostan2010complexity}.
The aim of this
work is to explain the gap for each instance we study.

% In fact, we propose a unified approach to prove better bounds
% (or to retrieve the best known bound)
% for these four problems at once.

\paragraph*{Contributions}

We consider representations of a rational matrix
$T \in k(x)^{n \times n}$ of the form
\begin{equation}
  \label{eq:realisation}
  T = W + XM^{-1}Y,
\end{equation}
where $W,X,M,Y$ are polynomial matrices of sizes $n \times n$,
$ n \times m$, $m \times m $, $m \times n$ respectively and $M$ is
non-singular.  Such a representation for $T$ is called a
\emph{realisation}~\cite{coppel1974matratfun} and originates in
control theory.
% This representation is convenient because it is often possible to
% find such a realisation for the matrices involved in our
% applications.
We focus here on strictly proper rational matrices: a rational matrix
is strictly proper if its limit when $x$ tends to infinity is $0$.  In
the case where $T$ is strictly proper, we have $ \det M \cdot T = N$
with $N$ a polynomial matrix of degree $< \deg (\det M)$.  We
establish a degree bound for the solution of
Problem~\ref{pb:lin_rel_pseudo_lin_map}, depending on $\deg (\det M)$.
% Note that the assumption that $T$ is strictly proper implies that
% $W = 0$ in~(\ref{eq:realisation}).
%
% Under a genericity assumption, we establish a degree bound for the
% solution of Problem~\ref{pb:lin_rel_pseudo_lin_map} depending on the
% realisation~(\ref{eq:realisation}) for $T$. Actually our degree
%% bound depends on the degree of $\det M$.
%
%
%
\begin{theorem}
  \label{thm:bound_realisation}
  Let $T \in k(x)^{n \times n}$ be a strictly proper rational matrix,
  $\theta = \DP_x + T$ and $a \in k[x]^n$.  Suppose we have a
  realisation~(\ref{eq:realisation}) of $T$. Let $\Delta = \det M$,
  $\delta = \deg \Delta$, and $d_a = \deg a$.  Then, there exists a
  solution
  $\eta = (\eta_0,\dots,\eta_\rho) \in k[x]^{\rho +1} \setminus \{ 0
  \}$ of Problem~\ref{pb:lin_rel_pseudo_lin_map} for $(T,a)$ 
  satisfying
  $\deg \eta_i \le \rho d_a + \rho \delta - \left( \rho(\rho+1)/2 -
    i\right)$.
\end{theorem}
\vspace{-10pt}
The coefficients of the solution of
Problem~\ref{pb:lin_rel_pseudo_lin_map} are thus bounded
by $O(nd_a  + n\delta)$.
If we choose the trivial realisation
$T = (\den \cdot T)(\den \cdot I_n)^{-1}$ where $\den$ is a least common
denominator of the entries of $T$ of degree $D$,
we retrieve the previously discussed bound in $O(n d_a + n^2 D)$, since we have $\delta = nD$.
In order to prove tight degree bounds,
we consider realisations of $T$ minimizing the degree of $\det M$, also known as
\emph{irreducible realisations}~\cite{coppel1974matratfun}.
In the four specific instances we focus on, we exhibit
a small realisation. Each time, the realisation is obtained by interpreting
the map $T$ as carrying out a reduction
% in a quotient vector space
modulo a polynomial or a relation.
%
%
% The direct approach in Section~\ref{sec:direct_bound}
% corresponds to the trivial realisation
% $T = $  In this case, we have $M = \den \cdot I_n$ and
% $\det M = \den^n$. Actually the degree gap occurs when there exist
% other realisations for $T$ with matrices $M$ with smaller
% determinants.  We say that the trivial realisation is \emph{reducible}
% in this case.  We prove the following result in
% Section~\ref{sec:deg_bound}.
%
%
% Suppose that we have a realisation of $T$ such that $L - \lambda = O(L/\rho)$, then the above
% bound becomes $O(\rho d_a + \rho \tilde L)$. Therefore if $D$ is a degree bound on both
% the numerator and the denominator of $T$, this has to be compared with the bound in $O(\rho d_a + \rho^2 D)$
% from Theorem~\ref{thm:direct_bound_cramer}. Thus if one is able to find a realisation of $T$ such that
% $\tilde L = o(\rho D)$ then Theorem~\ref{thm:bound_realisation} provides a better bound.
% \louis{We will do this for all our applications.}
%
We prove Theorem~\ref{thm:bound_realisation} in
Section~\ref{sec:deg_bound}. It follows from
structural properties on the denominators of minors of matrices of
rational functions, recalled in Section~\ref{sec:matrices_rational_functions}.
It is based on a linearisation
of the input matrix $T$ also known as
a \emph{state-space realisation} of $T$~\cite{kailath1980linear_systems}.
Besides, the strict properness assumption in Theorem~\ref{thm:bound_realisation}
seems to be only required because of the technical tools we use in the proofs.
This is at least suggested by our experiments.
% \louis{One word on the proof}

%\louis{What happen when $T$ is not strictly proper? What happen for Applications?}

In Section~\ref{sec:applications}, Theorem~\ref{thm:bound_realisation} is applied
to the four instances mentioned above.
For each instance, a degree bound is obtained and compared with the best known
bound in the literature. A summary is given in Table~\ref{tab:results_bound_applications}.
For the algorithm based on Hermite's reduction (row \textsf{Hermite}),
the degree bounds depend on $d_x$ and $d_y$ as defined before.
Similarly, for computing the differential resolvent of $P(x,y)$,
the degree bounds are given in terms of $d_x, d_y$ the degrees of
$P$ in $x$ and $y$ respectively (row \textsf{AlgeqtoDiffeq}). For these two applications,
our new bounds are valid under a genericity assumption that is made explicit in
Section~\ref{sec:hermite_red_creative_telesc}.
Finally, for LCLM and symmetric products of~$s$ linear
differential operators of order at most~$r$ and degree at most~$d$
(rows \textsf{LCLM} and \textsf{SymProd}), the bounds are expressed
in terms of the parameters $s, r$ and $d$. The new bounds
are established for operators that do not admit the point at infinity as an irregular
singularity, \ie it
is either ordinary or a regular singularity~\cite[\S20]{poole1936introduction}.
Note that the restrictions on our bounds (genericity, point at infinity not irregular)
are only necessary to fulfil the strict properness
assumption in Theorem~\ref{thm:bound_realisation}. A similar result
 without this assumption would yield bounds that
 hold for arbitrary input in our four instances.
 
\begin{table}
  \renewcommand{\arraystretch}{1.2}
  \begin{minipage}{1.0\linewidth}
    \begin{tabular}{c|cc|c}
    & \multicolumn{2}{c|}{Previous Bound} & New Bound \\
    \hline
    \textsf{Hermite} & $2d_y^2d_x + o(d_y^2d_x)$ &\cite{bostan2010complexity}&
    $2d_y^2d_x + o(d_y^2d_x)$ \savefootnote{genericity}{Under a genericity assumption (see Section~\ref{sec:hermite_red_creative_telesc})}\\
      \hline
    \textsf{AlgeqtoDiffeq} & $4d_y^2d_x + o(d_y^2d_x)$ &\cite{bostan2007algeqtodiffeq}&
    $2d_y^2d_x + o(d_y^2d_x)$ \repeatfootnote{genericity}\\
    \hline
    \textsf{LCLM} & $ds^2r + o(ds^2r)$ &\cite{bostan2012LCLM}&
   $ds^2r + o(ds^2r)$  \savefootnote{fuchs}{Under the
                                                      assumption that $x= \infty$
                                                      is not an irregular singularity}\\
      \hline
    \textsf{SymProd} &$O(dr^{2s})$ &\cite{kauers2014bounds_Dfinite_closure}&
    $O(dr^{2s-1}) $  \repeatfootnote{fuchs}
    \end{tabular}
  \end{minipage}
  \caption{Degree bounds for the four instances}
  \label{tab:results_bound_applications}
  \renewcommand{\arraystretch}{1}
 \end{table}
 \setlength{\textfloatsep}{-1pt}

\section{Results in specific instances}
\label{sec:applications}
The proof of our main result is postponed to Section~\ref{sec:deg_bound}.
For now we focus on its consequences for proving degree bounds in
the problems we have mentioned.
In each of the following instances, we exhibit
a realisation of the resulting matrix $T$.
We then study for which class
of inputs the matrix $T$ is strictly proper so that
Theorem~\ref{thm:bound_realisation} is applicable.

Let $\BK$ be a field and $s \ge 0$. We denote by
 $\BK[y]_{< s}$
 the ring of polynomials in $y$ of degree less than $s$ with coefficients in $\BK$.

\subsection{Hermite reduction for creative telescoping}
\label{sec:hermite_red_creative_telesc}

Let
$f = p/q$ be in $k(x,y) \setminus \{ 0 \}$ with $p$ and $q$ coprime in $k[x,y]$. 
For simplicity, we assume that $q$ is square-free with respect to $y$,
$\deg_y p < \deg_y q = d_y$ and $\deg_x p \le \deg_x q = d_x$.
We also denote $q_y = \DP_y(q)$ and
$q_x = \DP_x(q)$.
The goal is to compute a telescoper $L \in k[x]\left<\DP_x\right> \setminus \{ 0 \}$ for $f$, \ie
such that $L(f)= \DP_y(h)$ for $h \in k(x,y)$.

\subsubsection{Reduction to Problem~\ref{pb:lin_rel_pseudo_lin_map}}
\label{sec:hermite_red_pb1}

% The definite integral $F(x) = \int f(x,y) dy$ is
% known to be D-finite and we want to compute a differential equation
% satisfied by $F$.  To do so, one can look for a \emph{telescoper} for
% $f$, namely a linear differential operator
% $L \in k[x]\left<\DP_x\right>$ such that
% $L(x,\DP_x) \cdot f(x,y) = \DP_y \cdot h(x,y)$ for some rational
% function $h$.  Under reasonable assumptions on $h$, the operator $L$
% is an annihilator for $F$.

For any rational function $g \in k(x,y)$ whose denominator is a power
of $q$, we denote by $\herm(g)$ its Hermite
reduction~\cite{hermite1872integration}, namely the unique polynomial
$r(x,y) \in k(x)[y]$ such that $\deg_y(r) < \deg_y(q)$ and
$g = \DP_y(h) + r/q$ for $h \in k(x,y)$.  Let $\mathcal A$ be the
vector space $k(x)[y]_{< d_y}$.  It is known that $\herm$ defines a
$k(x)$-linear map from $k(x)[y,q^{-1}]$ to $\mathcal A$ and
$\herm(g) = 0$ if and only $g$ is a derivative (wrt.~$y$).  Moreover,
one can show~\cite{bostan2010complexity} that
$\herm(\DP_x\cdot g) = (\DP_x + T) \cdot \herm(g)$ where $T$ is the
$k(x)$-linear map:
\begin{align*}
  T \colon a \in \mathcal A \mapsto  - \herm(q_x a / q^2) \in \mathcal A.
\end{align*}
A telescoper $L$ for $f$ is such that $\herm(L\cdot f) = 0$ and thus by
linearity of $\herm$ it can be seen as a linear relation between
$\herm(\DP_x^i \cdot f)$ for $i \ge 0$.  Besides, by induction on
$i \ge 0$, we obtain $\herm(\DP_x^i \cdot f) = \theta^i(p)$ where $p = \herm(f)$
is the numerator of $f$ and $\theta$ the pseudo-linear map $\DP_x +
T$. Therefore, a minimal telescoper for $f$ can be obtained as the
solution of Problem~\ref{pb:lin_rel_pseudo_lin_map} on input
$(T,p)$. This is essentially the algorithm
in~\cite[Sec.~3.1]{bostan2010complexity}.

\subsubsection{Realisation of $T$}
\label{sec:hermite_realisation}

%Previously, the best known bound from
%the analysis of this algorithm was in~$O(d_y^3d_x)$~\cite[Lem.~18]{bostan2010complexity}.

In the basis $(1,y,\dots,y^{d_y-1})$, we consider $T$
as a matrix with rational function as coefficients, \ie
$T \in k(x)^{d_y \times d_y}$.
\begin{lemma}
  \label{lem:herm_red_XMm1Y}
  There exist polynomial matrices $X,M$ and $Y$, with $M$ invertible and
  $\deg \det M \le 2d_xd_y$ such that $T= XM^{-1}Y$.
\end{lemma}
\begin{proof}
  By analysing Hermite reduction for $q_x a /q^2$ and its associated
  linear system as
  in~\cite{horowitz1971algorithms,bostan2010complexity}, for all
  $a \in \mathcal A$ we have $q_x a /q^2 = \DP_y (A/q) + r/q$ for
  $A,r \in \mathcal A$ and $r = \herm(q_x a/q^2)$. This can be
  reformulated as
  \begin{align}
    \label{eq:hermite_reduction_system}
    q_xa = q \DP_y(A) - q_y A + qr,
  \end{align}
  just by multiplying the previous equality by $q^2$. Viewing $A$ and
  $r$ as undetermined polynomials in $\mathcal A$, $(A,r)$ is solution
  to the linear system
  \begin{align}
    \label{eq:hermite_reduction_system_matrix}
    \begin{bmatrix}
      M_1 & M_2
    \end{bmatrix}
    \begin{bmatrix}
      A \\ r
    \end{bmatrix} = q_x a,
  \end{align}
  where $M_1,M_2 \in k[x]^{2d_y \times d_y}$ of degree at most $d_x$.
  By~\cite[Lem.~6]{bostan2010complexity}, $M =
  \begin{bmatrix}
    M_1 & M_2
  \end{bmatrix}$ is invertible and by a direct expansion,
  its determinant has degree $ \le 2d_xd_y$.
  Let $Y$ denote the multiplication by $-q_x$, \ie $Y \colon a \mapsto -q_xa$.
  For all $a \in \mathcal A$, we have $Ta = XM^{-1}Y a$ where $X$ is simply the projection according to
  the last $d_y$ coordinates.
  % Thus $X,M,Y$ are three polynomial matrices.
  % Finally, $X$ has degree $0$, $Y$ has degree at most $d_x -1$ and each $\Lambda M^{-1}$ is a polynomial
  % matrix whose coefficents are minors of $M$ and thus have degree $\le 2d_xd_y - d_x =\lambda -d_x$. So
  % $\deg(\Lambda T) \le \lambda -1$ and this completes the proof.
\end{proof}

\subsubsection{Strict properness}
\label{sec:hermite_strict_properness}

It remains to study when the matrix $T$ is strictly proper in order to
apply Theorem~\ref{thm:bound_realisation} with the realisation of
Lemma~\ref{lem:herm_red_XMm1Y}.
Let $\Delta = \det M$. Since $T = XM^{-1}Y$ with $X$ of degree 0 and
$Y$ of degree at most $d_x -1$, it is sufficient that
$M^{-1} = O(1/x^{d_x})$ at infinity for $T$ to be strictly
proper. Note that $M^{-1} = 1/\Delta \cdot \hat M$ where $\hat M$ is the
adjugate matrix of $M$. So the entries of $\hat M$ are minors of $M$ of
dimension $2d_y -1$.  Besides, the entries of $M$ have degree at most
$d_x$, thus $\deg \Delta \le 2d_x d_y$ and the entries of $\hat M$ all
have degree at most $(2d_y-1)d_x$. Hence, it suffices that
$\deg \Delta = 2d_x d_y$ to have $M^{-1} = O(1/x^{d_x})$ and $T$
strictly proper.
%In particular,
%they are polynomials in $x$.
%Suppose that seen as a polynomial in
%$y$, all coefficients of $q$ (that are polynomials in $x$) have degree
%$d_x$, \ie we suppose that $q_{d_x,j} \neq 0$ for all $j$. 
\begin{lemma}
  \label{lem:expresion_Delta_Hermite_red}
  Up to a sign, $\Delta = \lc(q) \res_y(q,q_y)$, where $\lc(q)$ is the
  leading coefficient of $q$ in $y$ and $\res_y$ denotes the resultant
  wrt.~$y$.
\end{lemma}

\begin{proof}
  Let $\mathcal B$ be the vector space $k(x)[y]_{< 2d_y}$.  Recall
  that $M$ is the matrix of the system described
  in~(\ref{eq:hermite_reduction_system})
  and~(\ref{eq:hermite_reduction_system_matrix}).  We have $M =
  \begin{bmatrix}
    M_1 & M_2
  \end{bmatrix}$ with $M_2$ the matrix of the multiplication by $q$,
  \ie $M_2 \colon r \in \mathcal A \mapsto qr \in \mathcal B$, and
  $M_1$ the matrix of
  $A \in \mathcal A \mapsto (q\DP_y(A) -q_y A) \in \mathcal B$.
  Therefore, one can write $M_1 = M_2 \cdot U - M_3$ with $M_3$
  the matrix of the multiplication by $q_y$, \ie
  $M_3 \colon A \in \mathcal A \mapsto q_y A \in \mathcal B$,
  and $U$ the matrix of $A \in \mathcal A \mapsto \DP_y(A) \in \mathcal A$.
  Writing in the basis $(1,\dots,y^{d_y-1})$ of $\mathcal A$,
  \begin{align*}
    U = {
    \begin{bmatrix}
      0 & 1 & \\
        & 0 & 2 \\
        &   & \ddots & \ddots \\
        % &   & & 0 & d_y-1 \\
        % &&&&0
    \end{bmatrix}}.
  \end{align*}
  Hence every column of $M_2 \cdot U$ is a multiple of a column of
  $M_2$.  So,
  \begin{align*}
    \Delta = \det
    \begin{bmatrix}
      M_1 & M_2 
    \end{bmatrix} = \det
    \begin{bmatrix}
      -M_3 & M_2
    \end{bmatrix}.
  \end{align*}
  We just proved that, up to a sign, $\Delta$ is the determinant of
  $(A,r) \in \mathcal A^2 \mapsto q_y A + qr \in \mathcal B$.  The
  matrix of the previous map in the bases
  $((1,0), \dots, (y^{d_y-1},0), (0,1), \dots, (0,y^{d_y-1}))$ and
  $(1, \dots, y^{2d_y-1})$ is
  \begin{align*}
    \begin{bmatrix}
      S& * \\ 0& \lc(q)
    \end{bmatrix},
  \end{align*}
  where $S$ is a matrix representing the Sylvester map associated to
  $(q_y,q)$.  Finally, up to a sign, $\Delta = \lc(q) \res_y(q,q_y)$.
\end{proof}

As a consequence, $\Delta$ has maximal degree $2d_xd_y$ if and only if
$\deg(\lc(q)) = d_x$ and $ \res_y(q,q_y)$ has maximal degree
$(2d_y-1)d_x$.

\begin{lemma}
  \label{lem:maximal_degree_resultant}
  Let $q \in k[x,y]$ be square-free with respect to $y$ of degree
  $d_x$ and $d_y$ in $x$ and $y$ respectively.  The resultant
  $\res_y(q,q_y)$ has maximal degree $(2d_y-1)d_x$ iff the coefficient
  of degree $d_x$ in $x$ of $q$, which is a polynomial in $y$, has
  degree $d_y$ and is square-free.
\end{lemma}

\begin{proof}
  Let $S$ be the Sylvester matrix associated to $(q,q_y)$.  One can
  write $S = x^{d_x} S_0 + O(x^{d_{x}-1})$ with
  $S_0 \in k^{(2d_y-1)\times (2d_y-1)}$. So $\res_y(q,q_y)$ has
  maximal degree $(2d_y-1)d_x$ if and only if $\det S_0 \neq 0$.  Let
  $Q \in k[y]$ be the coefficient of $q$ of degree $d_x$ in $x$. $S_0$
  is the matrix of the map
  $(U,V) \in k[y]_{< d_y -1 } \times k[y]_{< d_y } \mapsto UQ +
  V \DP_y(Q) \in k[y]_{< 2d_y -1}$. For this map to be surjective, it is
  necessary that $\deg Q = d_y$, in which case it is the Sylvester map
  associated to $Q$ and $\DP_y(Q)$.  Hence, $S_0$ is invertible if and only
  if $Q$ has degree $d_y$ and is square-free.
\end{proof}

% One can
% write $M = x^{d_x} M_0 + O(x^{d_x-1})$ with
% $M_0 \in k^{2d_y \times 2 d_y}$. We have  $\det M_0 \neq 0$ (note
% that $\det M_0$ is a polynomial in the $q_{i,j}$'s), if and only if
% $\deg \det M = 2d_yd_x$.  And by a direct expansion of the
% determinant, any $(2d_y-1)$ minor of $M$ has degree at most
% $(2d_y-1) d_x$. So if % $q_{d_x,j} \neq 0$ for all $j$ and
% $\det M_0 \neq 0$, we have $M^{-1} = O(1/x^d)$ and $T$ is strictly
% proper. Let us finally exhibit one example where the above condition
% hold.

% The columns of $M_1$ in~(\ref{eq:hermite_reduction_system_matrix})
% can be seen as the sum of two components: one comes from
% $q\DP_y \cdot A$ and the other from $-q_y \cdot A$.
% So the first component in each column of $M_1$ is a multiple of
% a column of $M_2$ corresponding to the part $q\cdot r$.
% Hence the determinant of $M$ is the same (up to a sign) as the determinant of

% So $\det M$ has degree $2d_yd_x$ (\ie $\det M_0 \neq 0$)
% if and only if $q_{d_x,d_y} \neq 0$ and $\res_y(q,q_y)$ has maximal degree $(2d_y-1)d_x$.
% Let $q = x^{d_x} \cdot Q(y)$ with $Q = \sum_{j=0}^{d_y} q_{d_x,j} y^{d_y}$ and  $q_{d_x,d_y} \neq 0$.
% $\det S$ has maximal degree if and only if $Q(y)$ is square-free which is true
% generically over the choice of the $q_{i,j}$'s.
In summary, if the coefficient of degree $d_x$ in $x$ of $q$ has
degree~$d_y$ and is square-free, then $T$ is strictly proper.

\subsubsection{Degree bound}
\label{sec:Hermite_degree_bound}

\begin{theorem}
  \label{thm:herm_red_generic_bound}
  Let $f(x,y) = p/q \in k(x,y)$ with
  $q$ square-free with respect to $y$,
  $(d_x,d_y) = (\deg_x(q),\deg_y(q))$ and
  $\deg_x p \le d_x$, $\deg_y p < d_y$.
  If the coefficient of degree $d_x$ in $x$ of $q$ has degree
  $d_y$ and is square-free, then there exists a minimal telescoper
  for $f$ of order $r \le d_y$
  and degree at most $rd_x + 2rd_yd_x -1/2 \cdot r(r-1)$.
\end{theorem}

\begin{proof}
  Under this assumption,
  the matrix $T$ is strictly
  proper and by
  Lemma~\ref{lem:herm_red_XMm1Y}, there is a realisation
  $T = XM^{-1}Y$ with $\deg \det M \le 2 d_x d_y$.
  Thus by Theorem~\ref{thm:bound_realisation},
  the claimed degree bound holds. 
\end{proof}

% Hence from Lemma~\ref{lem:herm_red_XMm1Y},
% if Assumption~(\ref{hyp:assumption_Pi}) holds for $(T,p)$,
% one can apply Theorem~\ref{thm:bound_realisation}.
% Let $r \le d_y$ be the order of the minimal telescoper $L$
% for $f$, we deduce that the degree in $x$ of $L$ is bounded by
% $rd_x + r(\lambda-1) - (1/2 \cdot r(r-1) -r) \le rd_x + 2rd_yd_x -1/2 \cdot r(r-1)$.
In Theorem~\ref{thm:herm_red_generic_bound}, the condition
on the coefficient of degree $d_x$ in $x$ of $q$ holds
when the coefficients of $q$ do not belong to a certain nontrivial
hypersurface. In other words, this is a genericity assumption.
Thus, in the generic situation, we retrieve the best known bound
on the degree of a minimal telescoper for a bivariate rational
function~\cite[Thm.~25]{bostan2010complexity}, namely
a degree governed by $2d_y^2d_x$.
This bound was derived from the analysis of the Almkvist-Zeilberger
algorithm for rational functions~\cite{almkvist1990method}.
This is the first time a degree
bound in $ O(d_y^2 d_x)$ is derived directly from the Hermite reduction-based
algorithm.

\subsection{Differential resolvent}
\label{sec:algeq_to_diffeq}

Let $P(x,y) \in k[x,y]$ be square-free with respect to $y$ and let
$d_x$ and $d_y$ be the degrees in $x$ and $y$ of $P$. We write
$P_x = \DP_x(P)$ and $P_y = \DP_y(P)$.
In this section, $\mathcal A$ again denotes the vector space $k(x)[y]_{< d_y}$.

\subsubsection{Reduction to Problem~\ref{pb:lin_rel_pseudo_lin_map}}
\label{sec:algeqtodiffeq_red_pb1}

Following the lines
of~\cite[Sec.~2.1]{bostan2007algeqtodiffeq}, we first show that a
differential resolvent of $P$ can be computed as a solution of
Problem~\ref{pb:lin_rel_pseudo_lin_map}.  Let
$\alpha_1, \dots, \alpha_{d_y}$ be the roots of $P$ in the algebraic
closure of $k(x)$.  For any root $\alpha$ of $P$, by differentiating
$P(x,\alpha) =0$ and since $P$ is square-free, we have
$\DP_x (\alpha) = - P_x(x,\alpha) / P_y(x,\alpha)$.
Hence the derivation
$\DP_x$ uniquely extends to $k(x)(\alpha_1,\dots,\alpha_{d_y})$.  By
induction, for all $i\ge1$, there exists a polynomial
$C_i \in k[x,y]$ such that 
$\alpha^{(i)} = C_i(x,\alpha) / P_y(x,\alpha)^{2i-1}$,
for any root $\alpha$ of $P$.  Since $P$ and
$P_y$ are coprime, one can define $D_i(x,y)$  to be the polynomial
in $\mathcal A$
%is the image of $C_i/ P_y^{2i-1}$ modulo $P$. If one defines $D_0 = y$, the
%sequence $(D_i)_{i \ge 0}$ is
such that $D_i(x,\alpha) = \alpha^{(i)}$
for all roots $\alpha$ of $P$. Note that $D_0 = y$.
And thus it suffices to find the minimal
linear dependence between the $D_i$'s to compute the differential
resolvent of $P$.

Moreover, by differentiating $\alpha^{(i)} = D_i(x,\alpha)$, we obtain
\begin{align*}
  \alpha^{(i+1)} &= \DP_x (D_i(x,\alpha)) + \DP_x( \alpha) \cdot \DP_y (D_i(x,\alpha)) \\
                 &= \DP_x (D_i(x,\alpha)) - \frac{P_x(x,\alpha)}{P_y(x,\alpha)} \DP_y (D_i(x,\alpha)).
\end{align*}
Therefore $D_{i+1} = (\DP_x + T)\cdot D_i$, where $T$ is the $k(x)$-linear
map $T \colon a \mapsto - \DP_y(a) P_x /P_y ~\mathrm{mod}~P$ on
$\mathcal A$. Finally, the
differential resolvent of $P$ is the solution of
Problem~\ref{pb:lin_rel_pseudo_lin_map} for $(T,y)$. %with
%$e_2 = (0,1,0,\dots)^t$.

\subsubsection{Realisation of $T$}
\label{sec:algeqtodiffeq_realisation}
Fix the monomial basis $(1,y,\dots,y^{d_y-1})$ of
$\mathcal A$ and view $T$ as a matrix in $k(x)^{d_y \times d_y}$.

\begin{lemma}
  \label{lem:algeqtodiffeq_XMm1Y}
  There exist polynomial matrices $X, M$ and $Y$, with $M$ invertible
  and $\deg \det M \le (2d_y -1)d_x$ such that $T = XM^{-1}Y$.
  %let $\Lambda = \det M$, we have
  %$\lambda = \deg \Lambda \le (2d_y -1)d_x$ and
  %$\deg \Lambda T \le \lambda -1$.
\end{lemma}

\begin{proof}
  Recall that $T a$ is the solution $V$ of the following
  B{\'e}zout equation
  \begin{align*}
    -\DP_y(a)P_x = U \cdot P + V \cdot P_y.
  \end{align*}
  So if
  $Y \colon a \in \mathcal A \mapsto -\DP_y(a)P_x \in k(x)[y]_{<
    2d_y -1}$ and $X$ is the projection $(U,V) \mapsto V$ onto the
  last $d_y$ coordinates, we have $T = X M^{-1}Y$ where $M$ is the
  Sylvester map associated to $P$ and $P_y$ and $X,M,Y$ are three
  polynomial matrices. So $\det M$ is the resultant of $P$ and $P_y$
  (with respect to $y$) and it has degree at most $(2d_y-1)d_x$.
  % And finally $X$ has degree $0$, $Y$ has degree
  % $\le d_x -1$ and $\Delta M^{-1}$ has coefficients tht are minors of
  % $M$ and thus of degree at most $(2d_y-2)d_x$.  Therefore
  % $\deg \Lambda T \le \lambda -1$, and the result holds.
\end{proof}

\subsubsection{Strict properness}
\label{sec:algeqtodiffeq_strict_properness}

% The situation is really similar to that of
% Section~\ref{sec:hermite_red_creative_telesc}.
An analysis similar to that
of Section~\ref{sec:hermite_red_creative_telesc} shows that it
suffices that $M^{-1} = O(1/x^{d_x})$ for $T$ to be strictly
proper. This is in particular the case when $\det M$ has maximal
degree $(2d_y-1)d_x$, \ie, when $\res_y(P,P_y)$ has maximal degree.
By Lemma~\ref{lem:maximal_degree_resultant}, this property
holds when the coefficient of degree $d_x$ of~$P$ in $x$ has degree
$d_y$ and is square-free.

\subsubsection{Degree bound}
\label{sec:algeqtodiffeq_degree_bound}

\begin{theorem}
  \label{thm:algeqtodiffeq_generic_bound}
  Let $P \in k[x,y]$ be square-free with respect to $y$,
  $(d_x,d_y) = (\deg_x P, \deg_y P)$. If the coefficient
  of degree $d_x$ of $P$ in $x$ has degree
  $d_y$ and is square-free, then there exists a differential resolvent of
  $P$ of order $r \le d_y$ and degree at most
  $r(2d_y-1)d_x - 1/2 \cdot r(r-1)$.
\end{theorem}

\begin{proof}
  Under this condition, the matrix $T$ is strictly proper. Moreover, it has
  realisation $T = XM^{-1}Y$ with $\deg \det M \le (2d_y-1)d_x$ by
  Lemma~\ref{lem:algeqtodiffeq_XMm1Y}.  The claimed
  bound follows from Theorem~\ref{thm:bound_realisation}.
\end{proof}

% Let $r \le d_y$ be the
% order of the differential resolvent for $P$, then its degree in $x$
%  be bounded by
% $r (\lambda - 1) - (1/2\cdot r(r-1) -r) = r(2d_y-1)d_x -1/2 \cdot
% r(r-1)$.  
Again, the condition above is a genericity assumption on $P$.
Our bound is in $O(d_y^2 d_x)$ and governed by $2d_y^2 d_x$.
The previous best known bound behaves like 
$4d_y^2d_x$~\cite[Thm.~1]{bostan2007algeqtodiffeq} and its
proof is independent from  Cockle's algorithm.
So this is an improvement of the constant from $4$ to $2$ in
the generic situation.

Experimentally if we assume that $d_x = d_y = d$, we observe that the
output degree is bounded by $d(2d^2 - 3d +3)$. Our bound gives
$d(2d^2 - 3d/2 + 1/2)$ and the previous bound
is $d(4d^2 -11d/2 + 7/2)$~\cite{bostan2007algeqtodiffeq}.  This is
the first time that a bound with the first term $2d^3$ is obtained.

\subsection{Least common left multiple}
\label{sec:lclm}

Let $L_1, \dots, L_s$ be linear
differential operators in $k[x]\left< \DP_x \right>$ 
and let $L = \lclm(L_1,\dots,L_s)$ be a least common left multiple.
%
% A least common left multiple (LCLM)
% $L = \lclm(L_1,\dots,L_s)$ of $L_1,\dots,L_s$ is a minimal order
% differential operator whose solution set is the sum of solution sets
% of the $L_i$'s.
%
% Suppose that each $L_i$ has order ( degree in $\DP_x$) $ \le r$ and
% degree $\le d$ in $x$. 
Again, we apply Theorem~\ref{thm:bound_realisation} to
bound the degree of $L$ and compare with the existing
bound~\cite{bostan2012LCLM}.
For simplicity, we first assume $s=2$ and $L = \lclm(L_1,L_2)$ as the
approach will be easily generalised for arbitrary $s$.

\subsubsection{Reduction to Problem~\ref{pb:lin_rel_pseudo_lin_map}
and realisation}
\label{sec:lclm_red_pb1_realisation}

We follow the standard approach recalled in~\cite[Sec.~4.2.2]{bostan2012LCLM}
or~\cite[Algo.~4.27]{Kauers23}.
Let $r_1$,
$r_2$ be the respective orders of $L_1$ and $L_2$, and
$\alpha= \alpha_1 + \alpha_2$ be the sum of generic solutions of $L_1$ and
$L_2$ respectively.  The coefficients of $L$ are read from a linear relation
between the successive derivatives of $\alpha$.  These derivatives all
lie in the finite dimensional vector space spanned by
$A = (\alpha_1, \dots \alpha_1^{(r_1-1)},\alpha_2, \dots,
\alpha_2^{(r_2-1)})$.  Let
$L_1 = p_{1,r_1}\DP_x^{r_1} + \cdots + p_{1,0}$ and
$L_2 = p_{2,r_2}\DP_x^{r_2} + \cdots + p_{2,0}$ with $p_{1,r_1}$ and
$p_{2,r_2}$ nonzero.  By induction suppose that for $\ell \ge 0$, we have
written
$\alpha^{(\ell)} = \sum_{j=0}^{r_1 - 1 } a_j \alpha_1^{(j)} +
\sum_{j=0}^{r_2 -1 } b_j \alpha_2^{(j)}$ with $a_j,b_j \in k(x)$.
Then differentiation yields
\begin{align*}
  \alpha^{(\ell+1)} = \sum_{j=0}^{r_1 -1} \left( a_j' \alpha_1^{(j)} + a_j \alpha_1^{(j+1)} \right) +
  \sum_{j=0}^{r_2 -1} \left( b_j' \alpha_2^{(j)} + b_j \alpha_2^{(j+1)} \right).
\end{align*}
Finally, we use $L_1(\alpha_1) = L_2(\alpha_2) = 0$ to rewrite
$\alpha_1^{(r_1)}$ and $\alpha_2^{(r_2)}$ on the generating set~$A$. 
So if $V_\ell$ denotes
the vector of coefficients of $\alpha^{(\ell)}$ in $A$ for all $\ell$, we
obtain $V_{\ell+1} = (\DP_x+ T )\cdot V_\ell$, where
$T \in k(x)^{(r_1 + r_2) \times (r_1 + r_2)}$ is the block diagonal
matrix $\diag(C_1,C_2)$ with $C_i$ the companion matrix associated to
$L_i$, \ie $C_i \in k(x)^{r_i \times r_i}$ is companion with last
column $ -1/p_{i,r_i} \cdot
\begin{bmatrix}
  p_{i,0} &  \cdots & p_{i,r_i-1}
\end{bmatrix}^t$.
But such a block companion matrix $T$ is clearly not strictly proper
since some of its entries are $1$.

In order to be able to apply Theorem~\ref{thm:bound_realisation}
we propose
to express the problem in a different generating set.  Let $\mathfrak{d} = x \DP_x$,
also known as the Euler operator.  We choose to write the successive
derivatives of $\alpha$ as linear combinations of elements of
the new generating set 
$B = (\alpha_1, \dots, \mathfrak{d}^{r_1-1}\alpha_1, \alpha_2, \dots
,\mathfrak{d}^{r_2-1}\alpha_2)$.  Now if $W_\ell$ denotes the vector of
coefficients of $\alpha^{(\ell)}$ in $B$ for all~$\ell$, by differentiating
we get that $W_{\ell+1} = (\DP_x + T_\mathfrak{d})\cdot W_\ell$ with
$T_\mathfrak{d} \in k(x)^{(r_1 + r_2) \times (r_1 + r_2)}$ the matrix whose
columns are the derivatives of the elements of $B$ expressed on~$B$. 
Notice that for $i=1,2$ and $\ell \ge 0$, we have
$(\mathfrak{d}^\ell \alpha_i)' = 1/x \cdot \mathfrak{d}^{\ell+1} \alpha_i$.  Also, any
linear differential operator in $x$ and $\DP_x$ can be rewritten as a
linear differential operator of the same order in $x$ and $\mathfrak{d}$ with
the following degree bounds on the coefficients.
\begin{lemma}
  \label{lem:degree_conversion_diffeq_euler_diffeq}
  Let $L = \sum_{j=0}^r p_j(x) \DP_x^j \in k[x]\left<\DP_x\right>$ of
  order $r$ with $d_j = \deg p_j$. There exists an operator
  $\tilde L = \sum_{j=0}^r q_j(x) \mathfrak{d}^j \in
  k[x]\left<\mathfrak{d}\right>$ of order $r$ whose solution set is the same
  as $L$. Moreover, we have $\deg q_r = d_r$ and
  $\deg q_j \le r + \max_{\ell \ge j}( d_\ell -\ell)$ for all $j$.
\end{lemma}

\begin{proof}
  The result is a direct consequence of the well-known identity
  $x^{j}\DP_x^j =  \mathfrak{d}\cdots (\mathfrak{d} - j+1)$ for all
  $j \ge 0$~\cite[p.~26]{poole1936introduction}.
\end{proof}

\noindent By Lemma~\ref{lem:degree_conversion_diffeq_euler_diffeq},
$L_i = q_{i,r_i}(x) \mathfrak{d}^{r_i} + \cdots + q_{i,0}(x)$ with
$q_{i,r_i}$ nonzero for $i=1,2$.  Let $D_i$ be the companion matrix of
size $r_i$ whose last column is $-1/q_{i,r_i} \cdot
\begin{bmatrix}
  q_{i,0} & \cdots  & q_{i,r_i-1}
\end{bmatrix}^t $.
One can write
$ T_\mathfrak{d} = 1/x \cdot \diag(D_1,D_2)$.
Therefore $T_\mathfrak{d} = XM^{-1}$ with
\begin{align*}
  M = x \cdot \diag(1,\dots,1,-q_{1,r_1},1,\dots,1,-q_{2,r_2}),
\end{align*}
and $X$ the block diagonal matrix whose blocks are companion matrices
associated to $(q_{1,0},\dots,q_{1,r_1-1})$ and
$(q_{2,0},\dots, q_{2,r_2-1})$.

\subsubsection{Strict properness}
\label{sec:lclm_strict_properness}

The matrix $~T_{\mathfrak{d}}~$ is strictly proper
when $\diag(D_1,D_2) = O(1)$ and thus when
$\deg(q_{i,j}) \le \deg(q_{i,r_i})$ for all $j$ and $i=1,2$.  
%
%
% \begin{proof}
%   One can multiply $L$ by $x^r$ without modifying the solution set.
%   Let $\tilde L = x^r L = \sum_{j=0}^r \tilde p_j(x) x^j \DP_x^j$,
%   with $\tilde p_j = x^{r-j} p_j$ and $\deg \tilde p_j = d_j + r-j$.
%   By induction, one can show that
%   $x^{j}\DP_x^j = (\mathfrak{d} - j+1)\cdots \mathfrak{d}$. Therefore,
  % \begin{align*}
  %   \tilde L =  \sum_{j=0}^r \tilde p_j(x) (\mathfrak{d} - j+1)\cdots \mathfrak{d} = \sum_{j=0}^r q_j(x) \mathfrak{d}^j .
  % \end{align*}
  % By expanding the products $(\mathfrak{d} - j+1)\cdots \mathfrak{d}$, we deduce
  % that $\deg q_r = d_r$ and
  % $\deg q_j \le \max_{k \ge j}( \deg \tilde p_k)$. The result follows.
  % \end{proof}
%
By Lemma~\ref{lem:degree_conversion_diffeq_euler_diffeq}, we deduce
that $\deg(q_{i,j}) \le \deg(q_{i,r_i})$ when
$\deg(p_{i,j}) + r_i-j \le \deg(p_{i,r_i})$ for all $j$. 
This is equivalent to $x = \infty$ not being an irregular
singularity of
the equation defined by $L_i$~\cite[\S20]{poole1936introduction}.
%In particular, this condition is fulfilled when $L_i $ is a Fuchsian
%differential operator.
Hence, we conclude that if $x=\infty$ is not an irregular
singularity of $L_1$ nor of $L_2$, then $T_\mathfrak{d}$ is strictly proper.

% \begin{lemma}
  % \label{lem:fuchsian_strictly_proper}
  % If $L_1, L_2$ define two Fuchsian differential equations, then
  % $T_\mathfrak{d}$ is strictly proper.
%\end{lemma}
%\begin{proof}
%\end{proof}

\subsubsection{Degree bound}
\label{sec:lclm_degree_bound}

\begin{theorem}
  \label{thm:LCLM_fuchsian_degree_bound}
  Let $L_1, L_2 \in k[x]\left<\DP_x\right>$  be of
  respective orders $r_1$ and~$r_2$, and degrees at most $d$ in $x$.
  Suppose $x = \infty$ is not an irregular singularity of $L_1$ nor of $L_2$.
  Then there exists an LCLM $L$ of $L_1$ and $L_2$ of order $r \le r_1 + r_2$ and
  degree at most $r(r_1 + r_2 + 2d) -1/2 \cdot r (r-1)$.
\end{theorem}

\begin{proof}
  The coefficients of $L$ can be directly derived from the solution to
  Problem~\ref{pb:lin_rel_pseudo_lin_map} with input $(T_\mathfrak{d},a)$
  where
  \begin{align*}
    a= (
    \overbrace{1,0,\dots,0}^{r_1},\overbrace{1,0,\dots,0}^{r_2}
    )^t.
  \end{align*}
  Besides, $T_\mathfrak{d}$ is strictly
  proper by assumption.  Also, $T_\mathfrak{d} = XM^{-1}$ with
  $\deg \det M \le 2d + r_1 + r_2$.  Thus the result follows from
  Theorem~\ref{thm:bound_realisation}.
\end{proof}

This result generalizes with $s \ge 2$ differential equations
with no irregular singularity at infinity
and a matrix $T_\mathfrak{d}$ with $s$ different blocks.  Let $r_i$ denote
the order of $L_i$ for $1 \le i \le s$, $R = \sum_{i=1}^s r_i$ and $d$
be a bound on the degrees in $x$ of all the $L_i$'s.  A least common
left multiple $L = \lclm(L_1,\dots, L_s)$ can be computed as the
solution of Problem~\ref{pb:lin_rel_pseudo_lin_map} for
$T_\mathfrak{d} = XM^{-1} \in k(x)^{R \times R}$ and
\begin{align*}
  a=
  (
  \overbrace{1,0,\dots,0}^{r_1},\dots,\overbrace{1,0,\dots,0}^{r_s}
  )^t,
\end{align*}
where $X,M$ are polynomial matrices. This time $\Delta = \det M$ is
the product of $x^R$ and the leading coefficients of the $L_i$'s (up
to a sign), so $\deg \Delta \le sd + R$.  Hence by
Theorem~\ref{thm:bound_realisation}, the order of $L$ is less than~$R$
and its degree in $x$ is bounded by $R(sd + R)$.  This is slightly
larger than the bound $d(s(R+1)-R)$~\cite[Thm.~6]{bostan2012LCLM}, but
as $R \le sd$ since $x=\infty$ is not an irregular point,
the leading term of both bounds
behaves like $dsR \le ds^2 \max_i(r_i)$.

% This should be compared with the sharp degree bound established in
% And recall that for Fuchsian differential equations, we have
% $kd \ge r$ and thus both bounds behave like $dkr$ in this case.

\subsection{Symmetric product}
\label{sec:symmetric_product}

Let $L_1, L_2 \in k[x]\left<\DP_x\right>$ and $L = L_1 \otimes L_2$
be a symmetric product.  Let $r_i$ be the order of $L_i$
and $d_i$ be its degree for $i=1,2$. For $i=1,2$, we write
$L_i = \sum_{j=0}^{r_i} p_{i,j}\DP_x^j$ with $p_{i,j} \in k[x]$.  It
is known that the order of $L$ is at most $r_1 r_2$ and
the degree of $L$
in $x$ can be bounded by $r_1^2 r_2^2(d_1 + d_2)$~\cite[Thm.~8]{kauers2014bounds_Dfinite_closure}.
However, experiments
suggest that this bound is not tight since one expects the degree of
$L$ to be bounded by $(r_1r_2 - r_1 - r_2 + 2)(d_1r_2 + d_2r_1)$.

% Let
% us see what bound Theorem~\ref{thm:bound_realisation} can give on this
% example.

\subsubsection{Reduction to Problem~\ref{pb:lin_rel_pseudo_lin_map}}
\label{sec:sym_prod_red_pb1}

By Leibniz's rule, we express the successive derivatives of
$\alpha = \alpha_1 \alpha_2$ according to the elements
$(\mathfrak{d}^h \alpha_1 \mathfrak{d}^p \alpha_2)$ for $0 \le h < r_1$ and
$0 \le p < r_2$. We again use the Euler operator to satisfy the
strict properness assumption.  Let
$b_{h,p} = \mathfrak{d}^h \alpha_1 \mathfrak{d}^p \alpha_2$ for all $h,p$ and
suppose that we have
\begin{align*}
  \alpha^{(\ell)} = \sum_{h=0}^{r_1-1}  \sum_{p=0}^{r_2-1} e_{h,p}(x) b_{h,p},
\end{align*}
with $e_{h,p} \in k(x)$. We denote by $V_\ell \in k(x)^{r_1 r_2}$ the
vector of coefficients of $\alpha^{(\ell)}$ in this decomposition.
Differentiating gives
\begin{align*}
  \alpha^{(\ell+1)} = \sum_{h,p}\left( e_{h,p}' b_{h,p} + 1/x \cdot e_{h,p}(b_{h+1,p} + b_{h,p+1})\right).
\end{align*}
Next the relations $L_1(\alpha_1) = L_2(\alpha_2) = 0$ are used to
rewrite $b_{r_1,p}$ and $b_{h,r_2}$ in the expression above.  Thus,
one can write $V_{\ell+1} = (\DP_x + T) \cdot V_\ell$ where $T$ is
$k(x)$-linear and maps any $b_{h,p}$ for $0 \le h < r_1$ and
$0 \le p < r_2$ to $1/x \cdot (b_{h+1,p} + b_{h,p+1})$ rewritten
according to $(b_{h,p})$ for $0 \le h < r_1$ and $0 \le p < r_2$ using
$L_1$ and $L_2$.
For $i= 1,2$, let $L_i = \sum_{j=0}^{r_i}q_{i,j} \mathfrak{d}^i$ with
$q_{i,j} \in k[x]$.
For $t \ge 0$, $I_t$ denotes the identity matrix of size $t$.
In the basis
$(b_{0,0},\dots,b_{0,r_2-1},b_{1,0},\dots,b_{r_1-1,r_2-1})$, $T$ can
be seen as the matrix
\begin{align}
  \label{eq:symprod_T}
  T = 1/x \cdot(D_1 \otimes I_{r_2} + I_{r_1} \otimes D_2) \in k(x)^{r_1r_2
  \times r_1 r_2},
\end{align}
where $D_i$ denotes the same companion matrix as in
Section~\ref{sec:lclm} and $\otimes$ the Kronecker product.
Also, the coefficients of $L$ can be directly read from the solution to
Problem~\ref{pb:lin_rel_pseudo_lin_map} with input $(T,e_1)$ where
$e_1 = (1,0,\dots) \in k(x)^{r_1r_2}$. % with $T$ strictly proper.
Besides, if the order of $L$ is $r_1r_2$, we retrieve Amitsur's notion
of resultant of
the differential polynomials 
$L_1$ and $L_2$~\cite[\S~4]{amitsur1954diffpoly_divalgebra}.

\subsubsection{Realisation of $T$}
\label{sec:symprod_realisation}
We have
$T = (1/x \cdot D_1) \otimes I_{r_2} + I_{r_1} \otimes (1/x \cdot
D_2)$.
By Section~\ref{sec:lclm_red_pb1_realisation},
$1/x \cdot D_i$ can be written as $X_i M_i^{-1}$ with
$X_i \in k[x]^{r_i \times r_i}$ the companion matrix whose last column
is $
\begin{bmatrix}
  q_{i,0} & \cdots  & q_{i,r_i-1}
\end{bmatrix}^t$ and $M_i = \diag(x,\dots,x,-q_{i,r_i}x)
\in k[x]^{r_i \times r_i}$.
Moreover, we have $\det M_i = - x^{r_i}q_{i,r_i}$.
Next,
$(1/x \cdot D_1) \otimes I_{r_2} = (X_1M_1^{-1}) \otimes I_{r_2} =
(X_1 \otimes I_{r_2}) \cdot (M_1 \otimes I_{r_2})^{-1}$ and
similarly,
$I_{r_1} \otimes (1/x \cdot D_2) = (I_{r_1} \otimes X_2) \cdot
(I_{r_1} \otimes M_2)^{-1}$.  Now let $X =
\begin{bmatrix}
  X_1 \otimes I_{r_2} & I_{r_1} \otimes X_2
\end{bmatrix}$, $M = \diag(M_1 \otimes I_{r_2}, I_{r_1} \otimes M_2)$ and
$Y =
\begin{bmatrix}
  I_{r_1 r_2} & I_{r_1 r_2}
\end{bmatrix}^t$.
One can check that $T = XM^{-1}Y$,
$\Delta = \det M = (\det M_1)^{r_2} (\det M_2)^{r_1}$
and $\deg \Delta = 2r_1r_2 +d_1 r_2 + d_2 r_1$.

\subsubsection{Strict properness}
\label{sec:symprod_strict_prop}

From~(\ref{eq:symprod_T}),
it suffices that
$D_i = O(1)$ for $i=1,2$ for $T$ to be strictly proper. We conclude as in
Section~\ref{sec:lclm_strict_properness}:  $T$ is strictly
proper when both $L_1$ and
$L_2$ do not admit $x= \infty$ as an irregular singularity.

\subsubsection{Degree bound}
\label{sec:sym_prod_degree_bound}

\begin{theorem}
  \label{thm:sym_prod_fuchsian_degree_bound}
  Let $L_1, L_2 \in k[x]\left<\DP_x\right>$ be of
  respective orders $r_1$ and $r_2$ and degrees $d_1$ and $d_2$.
  Suppose $x = \infty$ is not an irregular singularity of $L_1$ nor of~$L_2$.
  Then
  there exists a symmetric product $L = L_1 \otimes L_2$ of
  order $r \le r_1r_2$ and degree at
  most $r(2r_1r_2 +d_1 r_2 + d_2 r_1) - 1/2 \cdot r(r-1)$.
\end{theorem}

\begin{proof}
  Under this assumption, $T$ is strictly proper. The result is
  implied by
  Theorem~\ref{thm:bound_realisation}
  with $\deg \Delta = 2r_1r_2 +d_1r_2 +d_2 r_1$.
\end{proof}

Our new bound, although restricted to the assumption that the point at infinity is
not irregular, leads to an
improvement of the bound in~\cite{kauers2014bounds_Dfinite_closure} by
one order of magnitude.  Recall that
$d_1 \ge r_1$ and $d_2 \ge r_2$ because of our assumption at infinity, so our bound is in
$O(d_1^2 d_2^2)$. This matches the asymptotic behavior we observe in
experiments. By contrast, the bound in~\cite{kauers2014bounds_Dfinite_closure}
behaves like $O(d_1^{2}d_2^{2}(d_1 +d_2))$ for this class of inputs.

Our approach also generalizes to the computation of a symmetric
product of $s \ge 2$ operators satisfying the same condition at infinity. Let
$L = L_1 \otimes \cdots \otimes L_s$ with $L_i$ of order $r_i$ and
degree at most $d$ for all $i$.
Let $R =  r_1 \cdots r_s$ and $r = \max_i(r_i)$.
By assumption, we have $r  \le d$.
Let
$\alpha = \alpha_1 \cdots \alpha_s$ where $L_i(\alpha_i) = 0$ for all
$i$.  By the Leibniz rule, one can express the
successive derivatives of~$\alpha$ as linear combinations of
$(\alpha_1^{(h_1)}\cdots \alpha_s^{(h_s)})$ for $0 \le h_i < r_i$,
$1 \le i \le s$.
Therefore $L$ can still be derived from a solution
of Problem~\ref{pb:lin_rel_pseudo_lin_map} with
$T \in k(x)^{R \times R}$. And one can write
$T = XM^{-1}Y$ generalizing the proof of
Theorem~\ref{thm:sym_prod_fuchsian_degree_bound}, with $\deg \det M
\le \sum_{i=1}^s (r_1 \cdots r_{i-1}(r_i + d)r_{i+1}\cdots r_s)$.
Then, $\deg \det M \le s(R + dr^{s-1})$ so $L$ is of order
at most $R$ and degree bounded by $sR(R + dr^{s-1}) = O(dr^{2s-1})$.
One should compare with the previous best known bound
$sR^2d = O(dr^{2s})$~\cite[Thm.~8]{kauers2014bounds_Dfinite_closure}.

\section{Direct bound by Cramer's rule}
\label{sec:direct_bound}

For the purpose of explaining our approach to show Theorem~\ref{thm:bound_realisation},
we start by proving a bound that is directly obtained from the
formulation of Problem~\ref{pb:lin_rel_pseudo_lin_map}: $\eta$ is a
solution of a linear system.  This is the bound one should expect by
adopting the most natural viewpoint on the input rational matrix
$T \in k(x)^{n\times n}$, \ie viewing $T$ as its numerator $N$ which
is a polynomial matrix divided by its denominator $\den \in k[x]$.

\begin{theorem}
  \label{thm:direct_bound_cramer}
  Let $T \in k(x)^{n \times n}$ with denominator $\den \in k[x]$, and
  $a \in k[x]^{n}$.  Let $d = \deg( \den)$, $D = \deg(\den \cdot T)$,
  $d_a = \deg a$ and $\tilde D = \max(d-1,D)$.  Then, there exists a
  solution $\eta = (\eta_0, \dots, \eta_\rho) \in k[x]^{\rho+1} \setminus \{ 0 \}$ 
  of
  Problem~\ref{pb:lin_rel_pseudo_lin_map} for $(T,a)$ of the
  form $\eta_i = \den^i p_i$ with $p_i \in k[x]$ of degree at most
  $\rho d_a + (1/2 \cdot \rho(\rho +1) - i) \tilde{D}$.
\end{theorem}
\begin{proof} %[Proof of Theorem~\ref{thm:direct_bound_cramer}]
  By induction on $i \ge 0$, we show that $\theta^i a = b_i / \den^i$
  with $b_i \in k[x]^n$ of degree at most $d_a + i \tilde{D}$. This is
  true for $i=0$ and
  $\theta^{i+1}a = \DP_x(b_i / \den^{i}) + T \cdot b_i/\den^i=
  1/\den^{i+1} (\den \cdot b_i' -i \den' \cdot b_i + (\den\cdot T) b_i)$. So we
  have $b_{i+1} = \den \cdot b_i' -i \den' \cdot b_i + (\den \cdot T) b_i$, and
  $\deg(b_{i+1}) \le \deg(b_i) + \tilde D$, thus the claim holds.
  
  Next consider the solution $\nu \in k(x)^\rho$ of the linear system
  $K \nu = - \theta^\rho a$ where $K =
  \begin{bmatrix}
    a & \theta a & \cdots & \theta^{\rho-1} a
  \end{bmatrix}\in k(x)^{n \times \rho}$.
  The solution~$\eta$ of Problem~\ref{pb:lin_rel_pseudo_lin_map} can
  be read from $\nu$ just by multiplying by the denominator of $\nu$.
  By definition of $\rho$, $\rk(K) = \rho$ so after a deletion of
  $(n-\rho)$ appropriate rows in the previous system, we obtain an
  equivalent square system $\tilde K \nu = c$, where $\tilde K$ and
  $c$ are obtained from $K$ and $-\theta^\rho a$ after the deletion of
  rows.

  By Cramer's rule, we have $\nu_i = \det K_i / \det \tilde K$ for all
  $0 \le i < \rho$ where $K_i$ is obtained from $\tilde K$ after
  replacing column $i+1$ by $c$.  Moreover by a direct expansion we
  have $\det \tilde K = p/\den^{\rho(\rho-1)/2}$ with $p \in k[x]$ of
  degree at most $\rho d_a + 1/2 \cdot\rho(\rho-1) \tilde D$ and
  $\det K_i = p_i/\den^{\rho(\rho+1)/2-i}$ with $p_i \in k[x]$ of
  degree at most $\rho d_a + (1/2 \cdot\rho(\rho+1)-i) \tilde D$.  So
  finally for $0 \le i < \rho$, $\nu_i = p_i / (\den^{\rho-i} p) $.

  And thus $\eta_i = \den^\rho p\cdot \nu_i = \den^i p_i $ for
  $0 \le i < \rho$ and $\eta_\rho = \den^\rho p$, whence the result.
\end{proof}

In summary, the above proof (which is similar to the approach 
in~\cite[\S2]{bostan2013cyclic})
boils down to the study of some of the
minors (up to sign) of a matrix $K^*$ whose columns are iterates of
$\theta$,
\begin{align*}
  K^* =
  \begin{bmatrix}
    a &  \cdots & \theta^{\rho-1}a & \theta^{\rho}a
  \end{bmatrix}.
\end{align*}
And we just bound the size of each minor by simply expanding the
determinant and bounding the degrees in every column. For proving
Theorem~\ref{thm:bound_realisation}, we first show that if the input
matrix $T$ is strictly proper and
has a realisation $T =W + XM^{-1}Y$ then the minors of the
matrix $K^*$ can always be written with denominator dividing
$(\det M)^\rho$ (see
Proposition~\ref{prop:det_den_pseudo_krylov_matrix}).  Now let $D$ be
a degree bound on the denominator of $T$. If we have a \emph{small}
realisation of $T$ meaning that $\deg \det M = O(D)$ then all these
minors are actually smaller than naive expansions suggest.

\section{Matrices of rational functions}
\label{sec:matrices_rational_functions}

We follow the lines
of~\cite{coppel1974matratfun,kailath1980linear_systems}
for the main definitions and results we need here.

\subsection{Linearisation of matrix fraction description}
\label{sec:MFD_linearisation}

Let $T \in k(x)^{n \times n}$ with a
realisation~(\ref{eq:realisation}). Such a
realisation for $T$ is not unique.  For example, for any non-singular
$m \times m$ polynomial matrices $D_1$ and $D_2$, one has the
realisation
\begin{align}
  \label{eq:reducible_realisation}
  T = W + XD_2(D_1MD_2)^{-1}D_1Y.
\end{align}
Recall that a matrix $U \in k[x]^{n\times n}$ is unimodular if its
determinant is a nonzero element in $k$. When $D_1$ or $D_2$ are not
unimodular % we feel that 
the
realisation~(\ref{eq:reducible_realisation}) can be \emph {reduced}.
%This is
%what be formalised with the notion of \emph{reducibility} of a
%realisation.

Let $A,B$ be $m \times p$ and $m \times h$ polynomial matrices. A
matrix $D \in k[x]^{m\times m}$ is a \emph{common left divisor} of $A$ and $B$
if there exist polynomial matrices $A_1, B_1$ such that $A = D A_1$
and $B = DB_1$. The matrices $A$ and $B$ are said to be \emph{left
relatively prime} if their only left common divisors are
unimodular. Similarly, one can define \emph{common right divisors} and 
\emph{relatively right prime} matrices simply by taking transposes in the above
definitions.  Now a realisation~(\ref{eq:realisation}) of $T$ is said
to be \emph{irreducible} if $M, Y$ are left relatively prime and $M, X$ are
right relatively prime. One can always reduce a given realisation to
an irreducible one by the following.
\begin{proposition} \emph{\cite[Thm.~7-9]{coppel1974matratfun}}
  \label{prop:reduced_realization_and_MFD}
  If $T$ has the realisation~(\ref{eq:realisation}) then there exist
  polynomial matrices $X_0,M_0,Y_0$ of size
  $n\times m, m \times m, m\times n$ respectively such that $T$ has
  the following irreducible realisation
  \begin{align}
    \label{eq:irreducible_realisation}
    T = W + X_0M_0^{-1}Y_0.
  \end{align}
  Moreover there exist $N, D \in k[x]^{n \times n}$ such that
  \begin{align}
    \label{eq:irreducible_rightMFD}
    T = ND^{-1}
  \end{align}
  is an irreducible realisation ($N$ and $D$ are right relatively
  prime).
\end{proposition}

\noindent A realisation of the form~(\ref{eq:irreducible_rightMFD}) is
also called an irreducible (right) \emph{matrix fraction description} (MFD)
for $T$.  The proof of
Proposition~\ref{prop:reduced_realization_and_MFD} is effective, in the sense
that irreducible
realisations~(\ref{eq:irreducible_realisation})
or~(\ref{eq:irreducible_rightMFD}) can be computed from an arbitrary realisation for
$T$.

Now given an MFD~(\ref{eq:irreducible_rightMFD}) for $T$, one can
\emph{linearise} it into a simple form called a \emph{state-space
realisation}~\cite{kailath1980linear_systems}.
\begin{proposition} \emph{\cite[Sec.~6.4]{kailath1980linear_systems}}
  \label{prop:state_space_realisation}
  Given a right MFD: $T = N D^{-1} \in k(x)^{n \times n}$,
  let $m = \deg \det D$, one can
  always construct a so-called \emph{state-space} realisation
  \begin{align}
    \label{eq:state-space_realisation}
    T(x) = W(x) + B(xI_m - A)^{-1}C, 
  \end{align}
  with $W \in k[x]^{n \times n}$ and $B,A,C$ matrices with entries
  in $k$ of size $n \times m, m \times m, m \times n$
  and $I_m$ the identity matrix of size $m$.
\end{proposition}
\noindent A state-space realisation~(\ref{eq:state-space_realisation})
is a particular example of a general
realisation~(\ref{eq:realisation}). When $T$ is strictly proper, one
has $W = 0$ in~(\ref{eq:state-space_realisation}) and the idea of
linearisation comes from the fact that we only have a polynomial
matrix of degree 1 involved.  We use a state-space realisation of
$T$ in our proof since this form is convenient when dealing with
differentiation. In particular, if $T$ is strictly proper, it has a simple
derivative $T' = -B(xI_m -A)^{-2}C$.

\subsection{Determinantal denominators}
\label{sec:determinantal_denominators}

We now recall results about the denominator of minors of
rational matrices.
For a rational matrix $R \in k(x)^{m \times p}$ and a positive integer~$\ell$, 
we denote by $\varphi_\ell(R)$ the monic least common
denominator of all minors of $R$ of size at most $\ell$.  We also set
$\varphi_0(R) = 1$. The polynomials $\varphi_\ell(R)$ are called the
\emph{determinantal denominators} of~$R$~\cite{coppel1974matratfun}.
Note that $\varphi_1(R)$ is the monic least common denominator of the
entries of $R$. So one can write $R = 1/\varphi_1(R) \cdot N$ with $N$
a polynomial matrix.
By definition, for $\ell \ge 0$, $\varphi_\ell(R)$ divides
$\varphi_{\ell +1}(R)$ and $\varphi_\ell(R) = \varphi_{\ell+1}(R)$ for
$\ell \ge \rk(R)$. Moreover by a direct expansion of the determinant,
$\varphi_\ell(R)$ divides
$\varphi_1(R)^\ell$.  We recall how determinantal
denominators behave with respect to sums, products and inverse of matrices.

\begin{proposition} \emph{\cite[Thm.~1-2]{coppel1974matratfun}}
  \label{prop:det_den_sum_prod}
  If a rational matrix $R$ is the sum $R_1 + R_2$ of two rational matrices
  $R_1, R_2$, then for all $\ell$, $\varphi_\ell(R)$ divides
  $\varphi_\ell(R_1)\varphi_\ell(R_2)$.  Moreover, if $\varphi_1(R_1)$
  and $\varphi_1(R_2)$ are coprime then equality holds.   Similarly
  if $R$ is the product $R_1R_2$ of two rational matrices $R_1, R_2$ then for all
  $\ell$, $\varphi_\ell(R)$ divides
  $\varphi_\ell(R_1)\varphi_\ell(R_2)$.
\end{proposition}

\begin{proposition} \emph{\cite[Thm.~4]{coppel1974matratfun}}
  \label{prop:det_den_inverse}
  Let $R$ be a non-singular $m \times m$ rational matrix.
  If $\det R = c \cdot  \alpha / \beta$ with $c \in k^*$,
  $\alpha, \beta \in k[x]$ monic,
  then $\beta \varphi_m(R^{-1}) = \alpha \varphi_m(R)$.
\end{proposition}

%\noindent The case of a
%realisation~(\ref{eq:realisation}) is important in this work.

\begin{proposition} \emph{\cite[Thm.~10]{coppel1974matratfun}}
  \label{prop:det_den_realisation}
  Let $T = W + XM^{-1}Y$ be a realisation of the
  form~(\ref{eq:realisation}).  For all $\ell \ge 0$,
  $\varphi_\ell(T)$ divides $\varphi_\ell(M^{-1})$.
  If the realisation is irreducible then equality holds.
\end{proposition}

% \begin{proof}
%   By \cite[Thm~4]{coppel1974matratfun}, we have
%   $\varphi_m(M^{-1}) =\det M$ and thus for all $\ell$,
%   $\varphi_\ell(M^{-1})$ divides $\det M$. We conclude using
%   Proposition~\ref{prop:det_den_sum_prod} since $\varphi_\ell(A) = 1$
%   for all $\ell \ge 0$ when $A$ is a polynomial matrix.
%\end{proof}

\begin{corollary}
  \label{cor:det_den_realisation_detM}
   Let $T = W + XM^{-1}Y$ be a realisation of the
  form~(\ref{eq:realisation}). For all $\ell \ge 0$,
  $\varphi_\ell(T)$ divides $\det M$.
\end{corollary}

\begin{proof}
  By Proposition~\ref{prop:det_den_realisation}, $\varphi_\ell(T)$
  divides $\varphi_\ell(M^{-1})$, which
  divides $\varphi_m(M^{-1})$ as $M$ is an $m \times m$ matrix.
  Next, by Proposition~\ref{prop:det_den_inverse},
  $\varphi_m(M^{-1})$ is equal to $\det M \cdot \varphi_m(M)$ up to
  a nonzero constant and
  $\varphi_m(M) = 1$ since $M$ is a polynomial matrix.
\end{proof}
\noindent Proposition~\ref{prop:det_den_realisation}
and Corollary~\ref{cor:det_den_realisation_detM} also
apply for MFD~(\ref{eq:irreducible_rightMFD}) or state-space
realisation~(\ref{eq:state-space_realisation})
as they are also examples of realisation~(\ref{eq:realisation}).

\section{Degree bound}
\label{sec:deg_bound}

In this section, we prove our main result
(Theorem~\ref{thm:bound_realisation}).
Let $T \in k(x)^{n\times n}$ be a strictly
proper matrix with realisation~(\ref{eq:realisation}).
We define the operator $\theta = \DP_x +T$
and take $a \in k[x]^n$.
We aim at
bounding the degrees in the coefficients of the solution of
Problem~\ref{pb:lin_rel_pseudo_lin_map} for $(T,a)$.
Let $\rho = \dim(\spn_{k(x)}(\theta^i a, i \ge 0))$ and
$\Delta = \det M$.

By Propositions~\ref{prop:reduced_realization_and_MFD} and~\ref{prop:state_space_realisation}
and the strict properness assumption,
one can write an irreducible
state-space realisation~(\ref{eq:state-space_realisation})
for $T$ in the form
\begin{align}
  \label{eq:state-space_realisation_strictly_proper}
  \tag{\ref{eq:state-space_realisation}$'$}
  T = B(xI_m - A)^{-1}C \eqcolon BL^{-1}C,
\end{align}
with $A,B,C$ scalar matrices.
%In what follows, we denote by $L$ the characteristic matrix $xI -A$.
Since the above state-space realisation is irreducible, we have
$\varphi_\ell(T) = \varphi_\ell(L^{-1})$ for all $\ell \ge 0$ by
Proposition~\ref{prop:det_den_realisation} and this divides $\Delta$
by Corollary~\ref{cor:det_den_realisation_detM}.
The degree bound is a consequence of the following result.
\begin{proposition}
  \label{prop:det_den_pseudo_krylov_matrix}
  Let $s_1 \le \cdots \le s_r$ be non-negative integers and consider the
  following matrix of iterates of $\theta$:
  \begin{align}
    \label{eq:pseudo_krylov_matrix}
    K =
    \begin{bmatrix}
      \theta^{s_1}a & \cdots & \theta^{s_r} a
    \end{bmatrix} \in k(x)^{n \times r}.
  \end{align}
  Then, $\varphi_\ell(K)$ divides $\Delta^{s_r}$ for all $\ell \ge 0$. 
\end{proposition}

Note that Proposition~\ref{prop:det_den_pseudo_krylov_matrix} is somehow
unexpected. Indeed, if one expands an arbitrary minor of $K$ of order $r$
by taking $\Delta$ as denominator of $T$, one expects this minor to have
denominator $\Delta^{s_1 + \cdots + s_r}$. It is not a priori clear why there should be
such simplifications in the denominator of the minors. Besides,
as for Theorem~\ref{thm:bound_realisation}, experiments suggest that
Proposition~\ref{prop:det_den_pseudo_krylov_matrix} still holds
if the matrix $T$ is no longer strictly proper.

\subsection{Proof of
  Proposition~\ref{prop:det_den_pseudo_krylov_matrix}}
\label{sec:pf_lemma_det_den_pseudo_krylov_matrices}

Let $\mathcal O = k(x)\left<\DP_x\right>$ be the ring of linear
differential polynomials with rational function
coefficients. Then, $\theta = I_n * \DP_x + T$
can be seen as a $n \times n$ matrix with
coefficients in
$\mathcal O$~\cite{jacobson1937pseudolineartransformations}.
Multiplication of
matrices of operators, denoted by $*$, is defined in the usual way by
\begin{align*}
  \DP_x * R = R  * \DP_x  + R',
\end{align*}
for any rational matrix $R \in k(x)^{n \times n}$. Next for any matrix
$U \in k(x)^{n \times p}$ with columns $u_1, \dots, u_p$,
let $\theta(U)$ be the result of applying
$\theta$ to each column of $U$, namely
\begin{align*}
  \theta(U) =
  \begin{bmatrix}
    \theta u_1 & \cdots & \theta u_p
  \end{bmatrix}.
\end{align*}
A simple reasoning by induction proves the following. 
\begin{lemma}
  \label{lem:binom_formula_theta}
  For all $s \ge 0$,
  $ \theta^s = \sum_{j=0}^s \binom{s}{j} \theta^{j}(I_n) * \DP_x^{s-j}
  \in \mathcal O^{n\times n}$.
\end{lemma}
%
% \begin{proof}
%   The result holds for $s=0$ as $\theta^0 = \theta^0(I) = I$.  Suppose
%   it holds for $s \ge 0$.  We have
  % \begin{align*}
  %   \theta^{s+1}&= (I\DP_x + T) \cdot \theta^s =  (I\DP_x + T)
  %   \cdot\left( \sum_{j=0}^s \binom{s}{j} \theta^{s-j}(I) \cdot \DP_x^j \right) \\
  %                   &= \sum_{j=0}^s \binom s j \left[ \theta^{s-j}(I) \DP_x  + \left(\theta^{s-j}(I)\right)'
  %                     + T \cdot \theta^{s-j}(I)  \right]\DP_x^j \\
  %                   &=\sum_{j=0}^s\binom s j\left[\theta^{s-j}(I) \DP_x^{j+1}+\theta^{s+1-j}(I) \DP_x^j \right]\\
  %                   &= \sum_{j=1}^{s+1} \binom s {j-1}  \theta^{s+1-j}(I) \DP_x^{j}
  %                   + \sum_{j=0}^{s} \binom s {j}  \theta^{s+1-j}(I) \DP_x^{j} \\
  %                   &= \theta^{s+1}(I) +\sum_{j=1}^{s} \binom {s+1} {j}\theta^{s+1-j}(I) \DP_x^{j} + \DP_x^{s+1}\\
  %                   &= \sum_{j=0}^{s+1} \binom {s+1} {j}  \theta^{s+1-j}(I) \DP_x^{j}.
  %                 \end{align*}
%   This proves the result by induction.
% \end{proof}
%
% By Lemma~\ref{lem:binom_formula_theta}, we have a formula for
% computing $\theta^s a$ for all $s$ involving the iterates of $\theta$
% applied on the identity matrix $\theta^j(I)$.  We give a useful
% formula for these matrices in the next lemmas.
We now compute a convenient expression for the matrices $\theta^j(I_n)$.
For $1 \le i \le j$, we
denote by $\Lambda_i^j$ the product
\begin{align*}
  \Lambda^j_i = \prod_{\ell = i}^{j-1} \left[(CB -\ell I_m)L^{-1}\right],
\end{align*}
with $L = xI_m - A$ and $A,B,C$ the scalar matrices
from~(\ref{eq:state-space_realisation_strictly_proper}).
We adopt the convention that $\Lambda_{j}^j = I_m$.
%We clearly have
%$\Lambda^j_i = (CB-iI)L^{-1} \Lambda^j_{i+1}$ for $i < j$.
\begin{lemma}
  \label{lem_formula_Lambda_ij}
  For $1 \le i \le j$, 
  $\Lambda^{j+1}_i - \left(\Lambda^j_i \right)' = (CB- i I_m) L^{-1}
  \Lambda^{j}_i $.
\end{lemma}
\begin{proof}
  Let $j \ge 1$, the result holds for $i =j$ since
  $\Lambda^{j+1}_j = (CB-jI_m)L^{-1}$ and $\Lambda_j^j = I_m$. Suppose it
  holds for $1 < i+1 \le j$ and recall that $(L^{-1})' = - L^{-2}$.
  Then, we obtain
  \begin{align*}
    & \Lambda^{j+1}_i - \left( \Lambda^j_i \right)' = (CB-i I_m)\left[ L^{-1}\Lambda^{j+1}_{i+1} -
      (L^{-1}\Lambda^j_{i+1})' \right] \\
    &= (CB-iI_m)L^{-1}\left[ \Lambda^{j+1}_{i+1} - (\Lambda^{j}_{i+1})'  + L^{-1}\Lambda^{j}_{i+1} \right].
    % &= (CB-iI)L^{-1}\left[ (CB-iI)L^{-1} \Lambda^{j}_{i+1} -
    %   L^{-1}\Lambda^{j}_{i+1}
    %   + (\Lambda^{j}_{i+1})'  \right] \\
    % &= (CB-iI)L^{-1}\left[ (CB-(i+1)I)L^{-1} \Lambda^{j}_{i+1}
    %   + (\Lambda^{j}_{i+1})'  \right] \\
    % &= (CB-iI)L^{-1} \Lambda^{j+1}_{i+1} = \Lambda^{j+1}_i.
  \end{align*}
  And by the induction hypothesis, the following concludes the proof,
  \begin{align*}
    &\Lambda^{j+1}_{i+1} - (\Lambda^{j}_{i+1})'  + L^{-1}\Lambda^{j}_{i+1}
      = (CB - (i+1)I_m)L^{-1} \Lambda^{j}_{i+1}  + L^{-1}\Lambda^{j}_{i+1} \\
    &= (CB -i I_m )L^{-1}\Lambda^{j}_{i+1} = \Lambda^j_i. \qedhere
  \end{align*}
\end{proof}

\begin{lemma}
  \label{lemma_thetaj_in_I}
  For $j \ge 1$, $\theta^j(I_n) = BL^{-1}\Lambda^j_1 C$.
\end{lemma}

\begin{proof}
  For $j = 1$, $\theta(I_n) = T = BL^{-1}C$, so the result
  holds.
  %For $j = 2$, we have $\theta^2(I) = T^2 + T' =
  %BL^{-1}CBL^{-1}C - B L ^{-2 } C = BL^{-1}(CB-I)L^{-1}C$.
  Suppose the result holds for $j \ge 1$. Then,
  $\theta^{j+1}(I_n) = (T + \DP_x)(\theta^j(I_n)) = BL^{-1}C B
  L^{-1}\Lambda_1^j C - BL^{-2}\Lambda_1^j C +
  BL^{-1}(\Lambda^j_1)'C$.  By Lemma~\ref{lem_formula_Lambda_ij},
  \begin{align*}
    &\theta^{j+1}(I_n) = BL^{-1}\left( (CB - I_m) L^{-1}\Lambda_1^j + (\Lambda^j_1)' \right)C
    = BL^{-1} \Lambda^{j+1}_1C. \qedhere
  \end{align*}
\end{proof}

%We are now finally ready to prove Lemma~\ref{prop:det_den_pseudo_krylov_matrix}. 
Now by Lemmas~\ref{lem:binom_formula_theta}
and~\ref{lemma_thetaj_in_I}, for $s \ge 0$, one can write
\begin{align*}
  \theta^s a = \sum_{j=0}^s  \theta^j(I_n) \binom s j a^{(s-j)} = a^{(s)}
  + BL^{-1}\left[ \sum_{j=1}^s \Lambda^j_1 C \binom s j a^{(s-j)} \right]. 
\end{align*}
For $i \ge 1$ and $s \ge 0$, let $u_i^s = C \binom s i a^{(s-i)}$ and
$\lambda_i^s = \sum_{j=i}^s \Lambda_i^j u_j^s$.  Note that $u_i^s = 0$
and $\lambda_i^s = 0$ when $i > s$.  Also $u_i^s$ is a vector of
polynomials.  So we have $\theta^s a = a^{(s)} + BL^{-1}\lambda_1^s$.
Hence the matrix $K$ in~(\ref{eq:pseudo_krylov_matrix})
can be written as
\begin{align*}
  K =
  \begin{bmatrix}
    a^{(s_1)} & \cdots & a^{(s_r)}
  \end{bmatrix}
  + BL^{-1} \mathcal{L}_1,
\end{align*}
with
$ \mathcal L_i = \begin{bmatrix} \lambda_i^{s_1} & \cdots &
  \lambda_i^{s_r}
\end{bmatrix}$ for $i \ge 1$.
In the above sum, the left term is a polynomial matrix.
So by Proposition~\ref{prop:det_den_sum_prod},
$\varphi_\ell(K)$ divides $\varphi_\ell(L^{-1})  \varphi_\ell(\mathcal L_1)$.
Therefore $\varphi_\ell(K)$ divides $\Delta    \varphi_\ell(\mathcal L_1)$.
It only remains to prove that $\varphi_\ell(\mathcal L_1)$ divides $\Delta^{s_r -1}$.
This is done in the following lemma, which
completes the proof of Proposition~\ref{prop:det_den_pseudo_krylov_matrix}.

\begin{lemma}
  \label{lem:det_den_Li}
  For $1 \le i \le s_r$ and $\ell \ge 0$, $\varphi_\ell(\mathcal L_i)$
  divides $\Delta^{s_r -i}$.
\end{lemma}

\begin{proof}
  $\mathcal L_{s_r}$ is a polynomial matrix, so
  $\varphi_\ell(\mathcal L_{s_r}) = 1$ for all $\ell \ge 0$.  Suppose
  the result holds for $1 < i+1 \le s_r$. For all $s \ge 0$, we have
  \begin{align*}
    \lambda_i^s = \sum_{j=i}^s \Lambda_i^j u_j^s = u_i^s +  \sum_{j=i+1}^s \Lambda_i^j u_j^s
    = u_i^s  + (CB - iI_m)L^{-1}\lambda_{i+1}^s.
  \end{align*}
  Thus, one can write
  \begin{align*}
    \mathcal L_i =
    \begin{bmatrix}
      u_i^{s_1} & \cdots & u_i^{s_r} 
    \end{bmatrix} + (CB-iI_m)L^{-1} \mathcal L_{i+1}.
  \end{align*}
  Since the left term is a polynomial matrix, by
  Proposition~\ref{prop:det_den_sum_prod},
  $\varphi_\ell(\mathcal L_i)$ divides
  $\varphi_\ell(L^{-1})\varphi_{\ell}(\mathcal L_{i+1})$.  Hence by
  induction it divides $\Delta^{s_r -i}$.
\end{proof}

\subsection{Proof of the degree bound}
\label{sec:pf_main_result}

We are now ready to establish the degree bound stated in Theorem~\ref{thm:bound_realisation}.
The beginning of the proof does not differ from the proof
of Theorem~\ref{thm:direct_bound_cramer},
except that we take $\Delta$ as a denominator of $T = XM^{-1}Y$.

\begin{proof}[Proof of Theorem~\ref{thm:bound_realisation}]
  We have $T = XM^{-1}Y$ so one can write $T$ with denominator
  $\Delta = \det M$.  And because $T$ is strictly proper, $\Delta T$
  is a polynomial matrix of degree $< \delta = \deg \det M$.  Thus
  similarly as in the proof of Theorem~\ref{thm:direct_bound_cramer}
  in Section~\ref{sec:direct_bound}, one can write $\theta^i a$ for
  all $i \ge 0$ as $b_i / \Delta^i$ with $b_i \in k[x]^n$ of degree at
  most $d_a + i(\delta-1)$.

  Let us again study the solution $\nu \in k(x)^\rho$ of the linear
  system $K \nu = - \theta^\rho a$ with $K =
  \begin{bmatrix}
    a & \theta a & \cdots & \theta^{\rho-1}a
  \end{bmatrix} \in k(x)^{n \times \rho}$.
  We also denote
  $K^* =
  \begin{bmatrix}
    K & \theta^{\rho } a
  \end{bmatrix} \in k(x)^{n \times \rho + 1}$.
  As in the proof of Theorem~\ref{thm:direct_bound_cramer}, applying Cramer's
  rule to this linear system shows that for all $0 \le i < \rho$,
  $\nu_i = \pm \mathfrak m_i / \mathfrak m  $ with $\mathfrak m_i$
  (resp.~$\mathfrak m$) a $\rho \times \rho$
  minor of $K^*$ (resp.~$K$).

  % We perform a valid deletion of $(n - \rho)$ rows in this system to make it square
  % and obtain $\tilde K \nu = c$ where $\tilde K$ and $c$ are obtained after deleting
  % the corresponding rows in $K$ and $- \theta^\rho a$.

  %Therefore by Cramer's rule, we can still write $\nu_i = \det K_i / \det \tilde K$
  %for $0 \le i< \rho$ where $K_i$ is obtained from $\tilde K$ after replacing column
  %$i+1$ by $c$.
  By a direct expansion, we get $\mathfrak m = p / \Delta^{\rho(\rho-1)/2}$
  with $p \in k[x]$ of degree at most $\rho d_a + 1/2\cdot \rho(\rho-1)(\delta-1)$.
  However, by Proposition~\ref{prop:det_den_pseudo_krylov_matrix},
  it can be written with
  denominator $\Delta^{\rho-1}$. Thus we have
  \begin{align*}
    \mathfrak m  = \frac p {\Delta^{\rho(\rho-1)/2}} = \frac q {\Delta^{\rho-1}}, 
  \end{align*}
  with $q = p / \Delta^{(\rho-1)(\rho-2)/2}$ \emph{a polynomial} of degree at most
  \begin{align*}
    \deg q \le \rho d_a + (\rho-1)\delta -1/2\cdot \rho(\rho-1).
  \end{align*}

  Similarly, $\mathfrak m_i$ can be written
  $p_i / \Delta^{\rho(\rho+1)/2 - i}$ with $p_i \in k[x]$ of degree
  at most $\rho d_a + (1/2 \cdot \rho(\rho+1) -i) (\delta-1)$.
  % But up to a sign, $\det K_i$ can be seen as a minor of the matrix
  % $K^* =
  % \begin{bmatrix}
  %   a & \cdots & \theta^{\rho-1}a & \theta^\rho a 
  %\end{bmatrix}$. Hence b
  Morever, by Proposition~\ref{prop:det_den_pseudo_krylov_matrix},
  it can be written with denominator $\Delta^\rho$.
  So we have
  \begin{align*}
    \mathfrak m_i = \frac{p_i}{\Delta^{\rho(\rho+1)/2 - i}} = \frac{q_i}{\Delta^{\rho}},
  \end{align*}
  with $q_i = p_i / \Delta^{\rho(\rho-1)/2 - i}$ \emph{a polynomial} of degree at most
  \begin{align*}
    \deg q_i \le \rho d_a + \rho \delta - (1/2 \cdot \rho(\rho+1) -i).
  \end{align*}

  Finally for $0 \le i < \rho$, $\nu_i = \pm q_i / (\Delta q)$.
  And thus $\eta_i = \Delta q \cdot \nu_i = \pm q_i$ for $0 \le i < \rho$
  and $\eta_\rho = \Delta q$, whence the result.
\end{proof}

\section{Perspectives}
\label{sec:perspectives}

Theorem~\ref{thm:bound_realisation} allows to establish precise degree
bounds in several concrete problems under some conditions (genericity,
point at infinity not irregular).  For future work, we aim at proving
a similar result %to Theorem~\ref{thm:bound_realisation}
without the
strict properness assumption.  This would yield bounds that hold
unconditionally in our four instances.  Also, there are other
problems that are instances of Problem~\ref{pb:lin_rel_pseudo_lin_map}
for which we could also derive bounds. This is for example the case
for the composition of algebraic functions and D-finite
functions~\cite{kauers2017bounds_algsubs}. Finally, the long-term goal is to
exploit the unified viewpoint we enlighted for the design of new
efficient algorithms. We notably intend to have algorithms
solving Problem~\ref{pb:lin_rel_pseudo_lin_map} whose
runtime is sensitive to the \emph{size} of the realisation of $T$
we have in input.

% Yet for each application, the strict properness assumption in
% Theorem~\ref{thm:bound_realisation} restricts the area where the
% degree bound holds.  Indeed, for the Hermite reduction-based
% algorithm or the differential resolvent, the bounds hold when the
% input satisfies a genericity assumption. For the LCLM or the
% symmetric product, we had to restrict ourselves to operators for
% which $x = \infty$ is not an irregular singularity.
% % to ensure the matrix $T$ to be strictly proper.
% However, the strict properness assumption seems to be only required
% because of the tools we use in the proof. This is at least suggested
% by our experiments. A similar result to
% Theorem~\ref{thm:bound_realisation} without this assumption would
% yield bounds that hold for arbitrary input in our four
% applications. % We leave this extension for future work.

\balance

\begin{acks}
  The author thanks Alin Bostan, Bruno Salvy and Gilles Villard for
  helpful discussions and comments. This work has been supported by
  the French–Austrian project EAGLES (ANR-22-CE91-0007 \& FWF
  I6130-N) %\href{https://anr.fr/en/funded-projects-and-impact/funded-projects/project/funded/project/b2d9d3668f92a3b9fbbf7866072501ef-49f46cc380/?tx_anrprojects_funded%5Bcontroller%5D=Funded&cHash=1baabb90cbff0f45adb1185571db7d62}
  and Agence nationale de la recherche (ANR) project NuSCAP
  (ANR-20-CE48-0014).
\end{acks}

%% ============================REFERENCES=============================================%%
	
\bibliographystyle{abbrv} \bibliography{biblio}

\end{document}

%% file: preamble.tex
\usepackage[utf8]{inputenc}

\usepackage{mathdots}
\usepackage{mathrsfs}

\usepackage{hyperref}
\hypersetup{
	colorlinks=true,     
	citecolor=blue,      % couleur des liens de citations
	linkcolor=black,      % couleur des liens internes
}

\usepackage{graphicx}
\usepackage{tikz}
\usepackage[ruled,linesnumbered]{algorithm2e}

\SetKwInOut{Input}{Input}
\SetKwInOut{Output}{Output}
\SetKwInOut{Oracle}{Oracle}

\usepackage{enumitem}

\newtheorem{problem}{Problem}
\newtheorem{theorem}{Theorem}
\newtheorem{proposition}{Proposition}

\newtheorem{lemma}{Lemma}
\newtheorem{corollary}{Corollary}

\newcommand{\savefootnote}[2]{\footnote{\label{#1}#2}}
\newcommand{\repeatfootnote}[1]{\textsuperscript{\ref{#1}}}

\makeatletter
\DeclareRobustCommand\onedot{\futurelet\@let@token\@onedot}
\def\@onedot{\ifx\@let@token.\else.\null\fi\xspace}
\def\ie{\emph{i.e}\onedot} 
\makeatother

\newcommand{\DP}{\partial}
\newcommand{\BK}{\mathbb{K}}

\newcommand{\den}{\mathrm{den}}

\DeclareMathOperator{\res}{Res}
\DeclareMathOperator{\spn}{span}

\DeclareMathOperator{\lc}{lc}

\DeclareMathOperator{\rk}{rk}
\DeclareMathOperator{\lclm}{LCLM}
\DeclareMathOperator{\diag}{Diag}
\DeclareMathOperator{\herm}{herm}